\documentclass[10pt, letterpaper]{article}
\usepackage[top = 0.85in, left = 2.25in, footskip = 0.75in]{geometry} 
\usepackage[shortlabels]{enumitem} 
\usepackage{amsmath, amsthm, amssymb}
\usepackage{mathtools} 
\usepackage{changepage}

\usepackage{textcomp, marvosym, adjustbox}

\usepackage{cite}

\usepackage{nameref, hyperref}

\usepackage[right]{lineno}

\usepackage[nopatch = eqnum]{microtype}
\DisableLigatures[f]{encoding = *, family = *}

\usepackage[table]{xcolor}

\usepackage{array}

\newcolumntype{+}{!{\vrule width 2pt}}

\newlength\savedwidth

\usepackage{ragged2e} 
\justifying 
\usepackage[aboveskip = 2pt, labelfont = bf, labelsep = period, singlelinecheck = off]{caption}

\makeatletter
\renewcommand{\@biblabel}[1]{\quad#1.}
\makeatother

\usepackage{lastpage, fancyhdr, graphicx}
\usepackage{epstopdf}

\fancyheadoffset[L]{2.25in}
\fancyfootoffset[L]{1.75in} 
\begin{document}
\vspace*{0.125in}


\begin{flushleft}

{\Large
    \textbf\newline{Information-Theoretical Measures for Developmental Cell-Fate Proportioning Processes}
}
\newline
\\
Michael Alexander Ramirez Sierra\textsuperscript{1,2,*},
Thomas R. Sokolowski\textsuperscript{1}
\\
\bigskip
\textbf{1} Frankfurt Institute for Advanced Studies (FIAS), Ruth-Moufang-Straße 1, 60438 Frankfurt am Main, Germany
\\
\textbf{2} Goethe-Universität Frankfurt am Main, Faculty of Computer Science and Mathematics, Robert-Mayer-Straße 10, 60054 Frankfurt am Main, Germany
\\
\bigskip

* ramirez-sierra@fias.uni-frankfurt.de

\end{flushleft}

\graphicspath{{./Manuscript_Graphics/}} 



\section*{Abstract}

Self-organization is a fundamental process of complex biological systems, particularly during the early stages of development. In the preimplantation mammalian embryo, blastocyst formation exemplifies a self-organized system, involving the correct spatio-temporal segregation of three distinct cell fates: trophectoderm (TE), epiblast (EPI), and primitive endoderm (PRE). Despite the significance of this class of processes, quantifying the information content of self-organizing patterning systems remains challenging due to the complexity and the qualitative diversity of developmental mechanisms. In this study, we applied a recently proposed information-theoretical framework which quantifies the self-organization potential of cell-fate patterning systems, employing a utility function that integrates (local) positional information and (global) correlational information extracted from developmental pattern ensembles. Specifically, we examined a stochastic and spatially resolved simulation model of EPI-PRE lineage proportioning, evaluating its information content across various simulation scenarios with different numbers of system cells. To overcome the computational challenges hindering the application of this novel framework, we developed a mathematical strategy that indirectly maps the low-dimensional cell-fate counting probability space to the high-dimensional cell-fate patterning probability space, enabling the estimation of self-organization potential for general cell-fate proportioning processes. Overall, this novel information-theoretical framework provides a promising, universal approach for quantifying self-organization in developmental biology. By formalizing measures of self-organization, the employed quantification framework offers a valuable tool for uncovering insights into the underlying principles of cell-fate specification and the emergence of complexity in early developmental systems.


\section*{Author summary}

We present the application of a newly proposed information-theoretical framework for quantifying the self-organization potential of cell-fate patterning processes. More crucially, given that the application of this quantification framework requires the estimation of a high-dimensional probability space, whose cardinality grows exponentially as a function of the number of system cells and the number of possible cell fates, we devise a generic mathematical and computational strategy to compute the information content of our simulated ICM (cell-fate proportioning) system. This formal-yet-tractable framework for quantifying developmental self-organization provides a universal approach that promises to enhance the discovery of optimal parameter regimes and the construction of highly detailed mechanistic models, offering a powerful tool for unveiling the common drivers of correct self-organized embryonic progression, irrespective of the underlying developmental mechanism or organism.


\section*{Motivation and Background}

The concept of self-organization is essential for the understanding of developmental biology systems \cite{bruckner_information_2024, plusa_common_2020, schroter_local_2023, zhu_principles_2020, morales_embryos_2021, beck_understanding_2024, ryan_lumen_2019}, but its quantification is complex in theory and challenging in practice because of the large variety of qualitatively different developmental mechanisms \cite{bruckner_information_2024}. For the preimplantation mammalian embryo, one of the prime examples of self-organized systems, its leading objective is the specification of three distinct cell fates: the trophectoderm (TE), the epiblast (EPI), and the primitive endoderm (PRE), thus giving rise to the blastocyst. The correct formation of the blastocyst, critical for the success of the post-implantation mammalian embryo, involves the spatio-temporal segregation of the two supporting tissue lineages (TE and PRE) alongside the early progenitor tissue lineage (EPI) of the embryo proper. This segregation must occur with well-balanced proportions of cell fates and within a relatively short developmental period \cite{plusa_common_2020, schroter_local_2023, zhu_principles_2020, ryan_lumen_2019}. The fascinating aspect about blastocyst formation is the coalescence of two seemingly contrasting and incongruent processes: a high self-renewal capacity, granting reliable maintenance of cell-potency, that coexists with a highly robust and reproducible cell-fate decision making \cite{plusa_common_2020}. More notably, both processes materialize despite the strongly stochastic conditions arising from highly variable intrinsic biochemical noise, extrinsic environmental fluctuations, as well as inhomogeneous spatial information mainly due to cell-cell signaling delays \cite{schroter_local_2023, morales_embryos_2021}. As described in detail in \cite{ramirez-sierra_ai-powered_2024}, embryonic cell-cell interactions, operating through a local cellular communication mechanism, orchestrate the global (blastocyst-wide) spatial patterns (i.e., emergence of complexity) that characterize this stochastic self-organizing system \cite{schroter_local_2023}.

Recent work has formalized the idea of ``positional information'' within the context of developmental systems, where the emergence of positional information is driven by maternal (or exogenous) morphogen gradients \cite{zagorski_decoding_2017, tkacik_many_2021, seyboldt_latent_2022, sokolowski_deriving_2023, bruckner_information_2024}. Yet, there is currently no universal (or generic) framework for quantifying the information content of self-organized cell-fate patterning processes such as blastocyst formation. For these self-organized systems, the input signals are transient and endogenous to the developing cell population; hence, a general quantification approach is required to estimate the information content of the system's output signals (cell-fate patterns) \cite{morales_embryos_2021, bruckner_information_2024}.

In response to this challenge, a novel theoretical study \cite{bruckner_information_2024} leverages information theory and interprets self-organization as a utility function ($\mathrm{U}$), proposing a measure that is the sum of two complementary quantities: positional information ($\mathrm{PI}$) and correlational information ($\mathrm{CI}$); i.e., $\mathrm{U} = \mathrm{PI} + \mathrm{CI}$. This information-theoretical measure mathematically encompasses the two sufficient and necessary criteria for self-organized cell-fate patterning processes \cite{bruckner_information_2024}: (1) developmental patterns emerge from initially homogeneous cell populations without exogenous (spatially coordinated) signaling factors, except for endogenous biochemical noise sources; (2) developmental patterns (cell-fate distributions) are reproducible across replicates of embryonic-like system experiments or simulations. Therefore, evaluating this utility function over some developmental pattern ensemble quantifies the information content or the self-organization potential of a given system \cite{bruckner_information_2024}.

More formally, a cell-fate patterning (i.e., developmental pattern) ensemble is a collection of realizations or replicates of a given stochastic process, which encapsulates the spatio-temporal dynamics governing the target behavior of some embryonic system. This abstraction is particularly suitable to embryonic or developmental systems, where correct spatial and temporal coordination is critical for pattern formation. Within this framework \cite{bruckner_information_2024}, $\mathrm{PI}$, the positional information, which quantifies the local spatial relationship between cell-fate identity and cellular positioning, can be decomposed as a linear combination of two entropies: $\mathrm{PI} = \mathrm{S}_{\mathrm{PAT}} - \mathrm{S}_{\mathrm{SCF}}$. $\mathrm{S}_{\mathrm{PAT}}$ is the ``patterning entropy'' which measures the cell-fate patterning diversity, and $\mathrm{S}_{\mathrm{SCF}}$ is the ``spatial-correlation-free entropy'' which measures the information content of an ancillary or implicit developmental system where there are no spatial correlations among cell-fate decisions. Similarly, $\mathrm{CI}$, the correlational information, which quantifies the global statistical structure underlying cell-fate identity and tissue patterning, can be decomposed as a linear combination of two entropies: $\mathrm{CI} = \mathrm{S}_{\mathrm{SCF}} - \mathrm{S}_{\mathrm{REP}}$. $\mathrm{S}_{\mathrm{REP}}$ is the ``reproducibility entropy'' which measures the cell-fate patterning distribution reproducibility over an ensemble of developmental patterns; estimating $\mathrm{S}_{\mathrm{REP}}$ requires ample knowledge of the inherent patterning probability space, which itself is remarkably challenging to gather in practice. A complete mathematical formalism for computing these quantities is detailed in the following section.

This quantification framework for developmental self-organization is promising as it proposes a generic utility function whose maximization, in conjunction with domain knowledge, would help identify optimal target behaviors for a given system within a set of degenerate solutions (cell-fate patterns). Such a normative theory would benefit, for instance, the exploration of the parameter spaces for the models of ICM specification (\cite{ramirez-sierra_ai-powered_2024}) and blastocyst patterning, fine-tuning and tightening their posterior parameter distribution approximations. However, for practical scenarios, this optimization is difficult to perform because the calculation of $\mathrm{S}_{\mathrm{REP}}$ involves the estimation of a high-dimensional probability space, whose cardinality grows exponentially as a function of the number of system cells ($n$) and the number of cell fates ($z$): $z^{n}$.

In this study, inspired by a recent collaboration with the lab of Gašper Tkačik at the Institute of Science and Technology Austria (ISTA - Klosterneuburg), we devise a mathematical and computational strategy to calculate $\mathrm{S}_{\mathrm{REP}}$ for the (EPI-PRE) cell-fate proportioning system introduced in \cite{ramirez-sierra_ai-powered_2024}: the Inferred-Theoretical Wild-Type ``ITWT'' system. This strategy is potentially applicable to similar self-organizing systems, and it can exploit experimental or synthetic (simulation) data; by formally relating the cell-fate patterning and cell-fate counting probability spaces, our strategy makes feasible the estimation of $\mathrm{S}_{\mathrm{REP}}$. We also demonstrate that, for an ideal or pure cell-fate proportioning process, $\mathrm{PI} = 0$ always, and $\mathrm{CI}$ is a strictly monotonically decreasing function of $n$. More importantly, we propose an alternative tractable measure that operates between the (high-dimensional) cell-fate ``patterning space'' and the (low-dimensional) cell-fate ``counting space'', facilitating the calculation of the reproducibility entropy for developmental pattern ensembles of cell-fate proportioning processes.


\section*{Methods and Results}

\subsection*{Basic definitions}

Let us fix some notation. $n$ is the number of system cells. $z$ is the number of cell fates. $\mathbb{N}$ is the set of nonnegative integers. $I_{0} = \{0, \ldots, n\} \subset \mathbb{N}$; $I_{1} = \{1, \ldots, n\} \subset \mathbb{N}$. $J_{1} = \{1, \ldots, z\} \subset \mathbb{N}$. The symbol $\vee$ denotes the logical OR conjunction. The symbol $\wedge$ denotes the logical AND conjunction. $N$ is a random variable uniformly distributed on $I_{0} \vee I_{1}$. $Z$ is a random variable uniformly distributed on $J_{1}$. $\mathcal{V}(\mathfrak{d})$ denotes a general vector space $\mathcal{V}$ of dimension $\mathfrak{d} \in \mathbb{N}$. $S$ denotes the Shannon entropy.

\theoremstyle{definition} \newtheorem{definition}{Definition}[section]

\begin{definition} \textbf{Decisioning Space} $\Omega(\vec{X})$. The decisioning space $\Omega(\vec{X})$ is the set of realizations ($d$-way tensors) $\vec{x}$ of a stochastic process $\Vec{X}$ such that each of their $n$ components is a realization $x_{\vec{v}}$ of the random variable $X_{\vec{v}}$ mapping to one of the fate indexes $j \in J_{1}$. The cell index $\vec{v}$ represents the coordinate vector of a cell in a $d$-dimensional vector space $\mathcal{V}(d)$; $d \in \{1, 2, 3\}$. \label{Definition_Decisioning_Space} \end{definition}

Although the decisioning space has an intuitive correspondence with the 3-dimensional physical space, instead we operate on an alternative space which is isomorphic to $\Omega(\vec{X})$ for simplicity; i.e., we simplify the cell indexes from coordinate vectors to straightforward positions (or locations) in $n$-element sequential arrangements.

\begin{definition} \textbf{Patterning Space} $\Omega(\vec{Z}) \cong \Omega(\vec{X})$. The patterning space $\Omega(\vec{Z})$ is the set of realizations (vectors) $\vec{z} = (z_{1}, \ldots, z_{i}, \ldots, z_{n})$ of a stochastic process $\vec{Z} = (Z_{1}, \ldots, Z_{i}, \ldots, Z_{n})$ such that each of their $n$ components is a realization $z_{i}$ of the random variable $Z_{i}$ mapping to one of the cell-fate indexes $j \in J_{1}$, where $n \leq \sum_{i=1}^{n}Z_{i} \leq nz$. The cell index $i$ represents the location of a cell in an $n$-dimensional vector space $\mathcal{V}(n)$. In symbols, \begin{equation*} \Omega(\vec{Z}) = \left\{\vec{z} \in \mathcal{V}(n) \Bigm\vert \vec{Z}:\mathcal{V}(n) \rightarrow J_{1}^{n} \quad \wedge \quad n \leq \sum_{i=1}^{n}Z_{i} \leq nz\right\} \textrm{.} \end{equation*} Here, \begin{equation*} J_{1}^{n} = \left\{(j_{1}, \ldots, j_{i}, \ldots, j_{n}) \in \mathbb{N}^{n} \Bigm\vert j_{i} \in J_{1} \quad \forall i \in I_{1}\right\} \textrm{.} \end{equation*} \label{Definition_Patterning_Space} \end{definition}

The elements of the patterning space are $n$-component sequences, and the cardinality of this set can be easily calculated by using the combinatorics concept of a permutation with repetition: each of the $n$ components can take any of the $z$ possible values, thus $\lvert\Omega(\vec{Z})\rvert = z^{n}$.

\begin{definition} \textbf{Patterning Probability Distribution} $P_{Z} \equiv P(Z)$. \begin{equation*} P_{Z} \doteq P(Z_{N} = Z \bigm\vert N \in I_{1} \wedge Z = j) \quad \forall j \in J_{1} \textrm{.} \end{equation*} In words, $P_{Z}$ is the overall probability of observing cell fate $Z = j \in J_{1}$ across the whole developmental pattern ensemble, regardless of cell position; $N \in I_{1}$. For convenience, $P_{Z}$ is equivalent to $P(Z)$. \label{Definition_Patterning_Probability_Distribution} \end{definition}

\begin{definition} \textbf{Spatial-Correlation-Free Probability Distribution} $P_{Z,N} \equiv P(Z,N)$. \begin{equation*} P_{Z,N} \doteq P(Z_{N} = Z \bigm\vert N = i \wedge Z = j) \quad \forall i \in I_{1} \wedge \forall j \in J_{1} \textrm{.} \end{equation*} In words, $P_{Z,N}$ is the probability of cell $N = i \in I_{1}$ adopting fate $Z = j \in J_{1}$ across the developmental pattern ensemble. For convenience, $P_{Z,N}$ is equivalent to $P(Z,N)$. \label{Definition_Spatial_Correlation_Free_Probability_Distribution} \end{definition}

\begin{definition} \textbf{Reproducibility Probability Distribution} $P_{\vec{Z}} \equiv P(\vec{Z})$. \begin{equation*} P_{\vec{Z}} \doteq P(\vec{Z} = \vec{z}) \quad \forall \vec{z} \in \Omega(\vec{Z}) \textrm{.} \end{equation*} In words, $P_{\vec{Z}}$ is the probability of observing realization (cell-fate pattern) $\vec{z} \in \Omega(\vec{Z})$. Recall that the vector $\vec{z}$ denotes the realization of the random vector $\vec{Z}$. For convenience, $P_{\vec{Z}}$ is equivalent to $P(\vec{Z})$. \label{Definition_Reproducibility_Probability_Distribution} \end{definition}

Using our notation, now we can formally define the information-theoretical measures proposed in the study by Brückner and Tkačik \cite{bruckner_information_2024}.

\begin{definition} \textbf{Patterning Entropy} $\mathrm{S}_{\mathrm{PAT}}$. \begin{equation*} \mathrm{S}_{\mathrm{PAT}} \doteq S\left[P_{Z}\right] = -\sum_{j=1}^{z} P(Z=j) \log_{2}[P(Z=j)] \textrm{.} \end{equation*} Recall that $\mathrm{S}_{\mathrm{PAT}}$ measures the cell-fate patterning diversity. Trivially, if $\exists ! j \in J_{1}$ such that $P(Z=j) = 1$, then $\min(\mathrm{S}_{\mathrm{PAT}}) = 0$. Similarly, if $P(Z=j) = 1/z$ for all $j \in J_{1}$, then $\max(\mathrm{S}_{\mathrm{PAT}}) = \log_{2}[z]$. \label{Definition_Patterning_Entropy} \end{definition}

\begin{definition} \textbf{Spatial-Correlation-Free Entropy} $\mathrm{S}_{\mathrm{SCF}}$. \begin{equation*} \mathrm{S}_{\mathrm{SCF}} \doteq \frac{1}{n} S\left[P_{Z,N}\right] = -\frac{1}{n} \sum_{i=1}^{n} \sum_{j=1}^{z} P(Z=j,N=i) \log_{2}[P(Z=j,N=i)] \textrm{.} \end{equation*} Recall that $\mathrm{S}_{\mathrm{SCF}}$ measures the information content of an ancillary or implicit developmental system where there are no spatial correlations among cell-fate decisions. \label{Definition_Spatial_Correlation_Free_Entropy} \end{definition}

\begin{definition} \textbf{Reproducibility Entropy} $\mathrm{S}_{\mathrm{REP}}$. \begin{equation*} \mathrm{S}_{\mathrm{REP}} \doteq \frac{1}{n} S\left[P_{\vec{Z}}\right] = -\frac{1}{n} \sum_{\vec{z} \in \Omega(\vec{Z})} P(\vec{Z}=\vec{z}) \log_{2}[P(\vec{Z}=\vec{z})] \textrm{.} \end{equation*} Recall that $\mathrm{S}_{\mathrm{REP}}$ measures the cell-fate patterning distribution reproducibility over an ensemble of developmental patterns. Trivially, if $\exists ! \vec{z} \in \Omega(\vec{Z})$ such that $P(\vec{Z}=\vec{z}) = 1$, then $\min(\mathrm{S}_{\mathrm{REP}}) = 0$. Similarly, if $P(\vec{Z}=\vec{z}) = 1/z^{n}$ for all $\vec{z} \in \Omega(\vec{Z})$, then $\max(\mathrm{S}_{\mathrm{REP}}) = \log_{2}[z]$. \label{Definition_Reproducibility_Entropy} \end{definition}

\begin{definition} \textbf{Utility Function} $\mathrm{U} = \mathrm{PI} + \mathrm{CI}$. The utility function $\mathrm{U}$ is the sum of the two entropies $\mathrm{PI} = \mathrm{S}_{\mathrm{PAT}} - \mathrm{S}_{\mathrm{SCF}}$ and $\mathrm{CI} = \mathrm{S}_{\mathrm{SCF}} - \mathrm{S}_{\mathrm{REP}}$. Consequently, \begin{equation*} \mathrm{U} = \mathrm{S}_{\mathrm{PAT}} - \mathrm{S}_{\mathrm{REP}} \textrm{.} \end{equation*} The optimization of $\mathrm{U}$ is hence a trade-off between maximizing $\mathrm{S}_{\mathrm{PAT}}$ and minimizing $\mathrm{S}_{\mathrm{REP}}$. Here, $\max(\mathrm{S}_{\mathrm{PAT}}) = \log_{2}[z]$ refers to a ``uniform patterning''; i.e., all (possible) cell fates are equiprobable. Likewise, $\min(\mathrm{S}_{\mathrm{REP}}) = 0$ refers to a ``perfect reproducibility''; i.e., all (observed) cell-fate patterns are identical. \label{Definition_Utility_Function} \end{definition}

However, because of the cardinality of the patterning space, the calculation of the reproducibility entropy is generally an unfeasible computational task, except for trivial scenarios. For estimating $\mathrm{S}_{\mathrm{REP}}$, here we employ a straightforward transformation from the cell-fate ``patterning space'' to the cell-fate ``counting space'', a partition of the patterning space using combinatorics, and a (pseudo-inverse) transformation mapping from counting vectors to sets of patterning vectors.

\subsection*{A transformation from patterning space to counting space}

Let us introduce some important and formal definitions.

\begin{definition} \textbf{Counting Space} $\Omega(\vec{N})$. The counting space $\Omega(\vec{N})$ is the set of realizations (vectors) $\vec{n} = (n_{1}, \ldots, n_{j}, \ldots, n_{z})$ of a stochastic process $\vec{N} = (N_{1}, \ldots, N_{j}, \ldots, N_{z})$ such that each of their $z$ components is a realization $n_{j}$ of the random variable $N_{j}$ mapping to one of the cell-count indexes $i \in I_{0}$, where $\sum_{j=1}^{z}N_{j} = n$. The fate index $j$ represents the location of a fate in a $z$-dimensional vector space $\mathcal{V}(z)$. In symbols, \begin{equation*} \Omega(\vec{N}) = \left\{\vec{n} \in \mathcal{V}(z) \Bigm\vert \vec{N}:\mathcal{V}(z) \rightarrow I_{0}^{z} \quad \wedge \quad \sum_{j=1}^{z}N_{j} = n\right\} \textrm{.} \end{equation*} Here, \begin{equation*} I_{0}^{z} = \left\{(i_{1}, \ldots, i_{j}, \ldots, i_{z}) \in \mathbb{N}^{z} \Bigm\vert i_{j} \in I_{0} \quad \forall j \in J_{1}\right\} \textrm{.} \end{equation*} \label{Definition_Counting_Space} \end{definition}

\begin{definition} \textbf{Transformation from Patterning Space to Counting Space} $T:\Omega(\vec{Z}) \rightarrow \Omega(\vec{N})$. The transformation $T$ is a surjective map between the patterning space $\Omega(\vec{Z})$ and the counting space $\Omega(\vec{N})$ such that \begin{equation*} \begin{gathered} T(\vec{Z}) = T[(Z_{1},\ldots,Z_{i},\ldots,Z_{n})] = (N_{1},\ldots,N_{j},\ldots,N_{z}) = \vec{N} = \\ = \left(\frac{1}{1}\sum_{i_{1}\in\pi_{1}}Z_{i_{1}}, \ldots, \frac{1}{j}\sum_{i_{j}\in\pi_{j}}Z_{i_{j}}, \ldots, \frac{1}{z}\sum_{i_{z}\in\pi_{z}}Z_{i_{z}}\right) \textrm{.} \end{gathered} \end{equation*} By construction, \begin{equation*} \pi_{j} \doteq \left\{i \in I_{1} \bigm\vert Z_{i} = j\right\} \quad \forall j \in J_{1} \textrm{.} \end{equation*} Here, the set $\pi = \{\pi_{1},\ldots,\pi_{j},\ldots,\pi_{z}\}$ is a partition (set of blocks) of the set $I_{1} = \{1,\ldots,n\}$; i.e., $\bigcup_{j=1}^{z}\pi_{j} = I_{1}$, and $\pi_{j_{1}} \bigcap \pi_{j_{2}} = \varnothing$ for all $j_{1} \neq j_{2} \in J_{1}$. Strictly, if $\exists j \in J_{1}$ such that $\pi_{j} = \varnothing$, then $\frac{1}{j}\sum_{i_{j}\in\pi_{j}}Z_{i_{j}} \doteq 0$. \label{Definition_Transformation} \end{definition}

We now consider the particular case of a binary cell-fate decision process, where a cell exhibiting an undifferentiated (progenitor) state can adopt either of two competing (progeny) states; thus, there are only three possible cell fates (including the undifferentiated cell fate), which is applicable to the ITWT system (see \cite{ramirez-sierra_ai-powered_2024}). Although it may seem restrictive to consider only such a binary cell-fate decision process, any complex decision can be decomposed as a sequence of (simple) binary decisions.

\newtheorem{proposition}{Proposition}[section]

\begin{proposition} \textbf{Counting Space Cardinality} $\lvert\Omega(\vec{N})\rvert$. Let $n \in \mathbb{N}$, and let $z = 3$. \begin{equation*} \lvert\Omega(\vec{N})\rvert = \binom{z+n-1}{n} = \frac{(z+n-1)!}{n!(z-1)!} = \frac{(n+1)(n+2)}{2} \textrm{.} \end{equation*} \label{Proposition_Counting_Space_Cardinality} \end{proposition}

\begin{proof} Without loss of generality, let us count the number of vectors $\vec{n} = (n_{1}, n_{2}, n_{3}) \in \Omega(\vec{N})$ for which $n_{2}+n_{3} \in \{0, 1, \ldots, n-1, n\}$. Trivially, $n_{1} \in \{n, n-1, \ldots, 1, 0\}$ given that $z = 3$ and ${\Vert\vec{n}\Vert}_{1} = n$. Consequently, there is one vector $\vec{n}$ such that $n_{2}+n_{3} = 0$ ($n_{1} = n$), there are two vectors $\vec{n}$ such that $n_{2}+n_{3} = 1$ ($n_{1} = n-1$), \ldots, there are $n$ vectors $\vec{n}$ such that $n_{2}+n_{3} = n-1$ ($n_{1} = 1$), and there are $n+1$ vectors $\vec{n}$ such that $n_{2}+n_{3} = n$ ($n_{1} = 0$). This recursive relation can be formalized via the concepts of ``multiset'' and ``multichoose coefficient'' (which is akin to the multinomial coefficient): \begin{equation*} \begin{gathered} \sum_{i=1}^{n+1} i = \binom{z+n-1}{n} = \frac{(z+n-1)!}{n!(z-1)!} = \frac{(n+2)!}{n!2!} = \\ = \frac{(n+2)(n+1)n!}{n!2!} = \frac{(n+1)(n+2)}{2} = \lvert\Omega(\vec{N})\rvert \textrm{.} \end{gathered} \end{equation*} \end{proof}

In Appendix Fig~\ref{ApexC_Fig1}, we illustrate a direct implication of Proposition \ref{Proposition_Counting_Space_Cardinality}; namely, if $n \gg 1$ then $((n+1)(n+2))/2 \ll z^{n}$. This detail is crucial because it indicates that, together with some extra assumptions, it is computationally feasible to tabulate all the realizable counting vectors. Following this tabulation, we can estimate counting vector probabilities (e.g., using experimental or simulation data) and indirectly map them to patterning vector probabilities, thus obtaining an approximation for the probability space of $\Omega(\vec{Z})$ from the available empirical knowledge about the probability space of $\Omega(\vec{N})$. We explicitly describe this strategy for estimating $\mathrm{S}_{\mathrm{REP}}$ over the next subsections.

\subsection*{A pseudo-inverse transformation from counting space to patterning space}

The first step of our $\mathrm{S}_{\mathrm{REP}}$ estimation strategy is the creation of a mapping from counting space to patterning space. Since $\lvert\Omega(\vec{N})\rvert < \lvert\Omega(\vec{Z})\rvert$, except for trivial cases, this mapping cannot be a bijection between these two spaces. Therefore, this mapping must relate (the probability of) a given counting vector to (the probability of) the unique set of patterning vectors giving rise to it under $T$ (Definition \ref{Definition_Transformation}). For a cell-fate proportioning process the main objective is to generate reproducible cell-fate counting distributions, irrespective of the resultant cell-fate spatial arrangements. Accordingly, a suitable assumption considers that all the elements belonging to the unique set of patterning vectors giving rise to a given counting vector under $T$ are equiprobable (uniform). Although the previous assumption is natural under this perspective, it is also flexible, allowing for distinct probability rules. However, in the following, we use the uniformity assumption to formalize our strategy.

The symbol $\mathcal{P}[\Omega(\vec{Z})]$ denotes the power set (set of all subsets) of the patterning space $\Omega(\vec{Z})$; i.e., $\mathcal{P}[\Omega(\vec{Z})] = \{\Omega \bigm\vert \Omega \subset \Omega(\vec{Z})\}$.

\begin{definition} \textbf{Pseudo-Inverse Transformation from Counting Space to Patterning Space} $M:\Omega(\vec{N}) \rightarrow \mathcal{P}[\Omega(\vec{Z})]$. The pseudo-inverse transformation $M$ is an injective map between the counting space $\Omega(\vec{N})$ and the power set of the patterning space $\mathcal{P}[\Omega(\vec{Z})]$ such that \begin{equation*} \begin{gathered} M(\vec{N}) \equiv \overset{M}{\Omega}(\vec{Z}\vert\vec{N}) \doteq \left\{\vec{z} \in \Omega(\vec{Z}) \bigm\vert T(\vec{z}) = \vec{n}\right\} \in \mathcal{P}[\Omega(\vec{Z})] \\ \forall \vec{n} \in \Omega(\vec{N}) \textrm{.} \end{gathered} \end{equation*} By construction, \begin{equation*} P\left[\vec{Z}\in\overset{M}{\Omega}(\vec{Z}\vert\vec{N}=\vec{n})\right] \doteq P(\vec{N}=\vec{n}) \textrm{.} \end{equation*} Moreover, \begin{equation*} P\left[\vec{Z}=\vec{z} \Bigm\vert \vec{z} \in M(\vec{N}=\vec{n})\right] \doteq \frac{P(\vec{N}=\vec{n})}{\left\lvert \overset{M}{\Omega}(\vec{Z}\vert\vec{N}=\vec{n}) \right\rvert} = P(\vec{N}=\vec{n})\frac{\prod_{j=1}^{z}n_{j}!}{n!} \textrm{.} \end{equation*} In other words, we assume that all the elements $\vec{z}$ belonging to the set $\overset{M}{\Omega}(\vec{Z}\vert\vec{N}) \equiv M(\vec{N})$ are equiprobable, for any $\vec{n} \in \Omega(\vec{N})$. \label{Definition_Pseudo_Inverse_Transformation} \end{definition}

The cardinality of the set $\overset{M}{\Omega}(\vec{Z}\vert\vec{N})$, for any $\vec{n} \in \Omega(\vec{N})$, \begin{equation*} \left\lvert \overset{M}{\Omega}(\vec{Z}\vert\vec{N}=\vec{n}) \right\rvert = \frac{n!}{n_{1}! \cdots n_{j}! \cdots n_{z}!} \end{equation*} is a straightforward application of the multinomial coefficient formula because, more formally, we are given an $n$-component sequence $\vec{z}$ and its associated $z$-component sequence $\vec{n}$ such that $T(\vec{z})=\vec{n}$ and ${\Vert\vec{n}\Vert}_{1} = \sum_{j=1}^{z}n_{j} = n$. Thus, we can easily enumerate the vectors $\vec{z}$ associated with a given vector $\vec{n}$ by considering partitions $\pi$ of an arbitrary $n$-element set into $z$ blocks, where the block $\pi_{j}$ contains exactly $n_{j}$ elements for all $j \in J_{1}$.

To properly frame our approach, we now consider several types of patterning spaces. In this way, we can meaningfully compare the proposed information-theoretical measures by creating multiple guiding references. The most important (guiding reference) type of patterning space, the ideal patterning space, combines the uniformity assumption underlying the pseudo-inverse transformation $M$ (Definition \ref{Definition_Pseudo_Inverse_Transformation}) with the assumption of a perfect counting reproducibility. Hence, the ideal patterning space represents the best possible case, providing an important performance benchmark for any given (empirical) cell-fate proportioning process. We formalize these patterning space types and several related ideas next.

Hereafter, $\overset{\star}{\vec{n}}=(\overset{\star}{n}_{1},\ldots,\overset{\star}{n}_{j},\ldots,\overset{\star}{n}_{z})$ denotes the ideal counting vector: a vector $\vec{n}$ belonging to the counting space $\Omega(\vec{N})$ which represents the ideal or perfect target cell counts for all the possible fates; e.g., for the ITWT system, $\overset{\star}{\vec{n}}=(\overset{\star}{n}_{\textrm{UND}},\overset{\star}{n}_{\textrm{EPI}},\overset{\star}{n}_{\textrm{PRE}})=(0,\frac{2}{5}n,\frac{3}{5}n)$. Recall that UND, EPI, and PRE refer to undifferentiated, epiblast, and primitive endoderm cell fates, respectively.

\begin{definition} \textbf{Ideal Patterning Space} $\overset{\star}{\Omega}(\vec{Z})$. Let \begin{equation*} \overset{M}{\Omega}\left(\vec{Z}\Bigm\vert\vec{N}=\overset{\star}{\vec{n}}\right) \doteq \left\{\vec{z} \in \Omega(\vec{Z}) \Bigm\vert T(\vec{z}) = \overset{\star}{\vec{n}}\right\} \textrm{.} \end{equation*} Consequently, by Definition \ref{Definition_Pseudo_Inverse_Transformation}, \begin{equation*} \left\lvert \overset{M}{\Omega}\left(\vec{Z}\Bigm\vert\vec{N}=\overset{\star}{\vec{n}}\right) \right\rvert = \frac{n!}{\overset{\star}{n}_{1}! \cdots \overset{\star}{n}_{j}! \cdots \overset{\star}{n}_{z}!} \textrm{.} \end{equation*} The ideal patterning space $\overset{\star}{\Omega}(\vec{Z})$ is a restriction of the patterning space $\Omega(\vec{Z})$ such that \begin{equation*} P\left[\vec{Z} \in \overset{M}{\Omega}\left(\vec{Z}\Bigm\vert\vec{N}=\overset{\star}{\vec{n}}\right)\right] = 1 \textrm{.} \end{equation*} Trivially, \begin{equation*} P\left[\vec{Z} \in \Omega(\vec{Z}) \setminus \overset{M}{\Omega}\left(\vec{Z}\Bigm\vert\vec{N}=\overset{\star}{\vec{n}}\right)\right] = 0 \textrm{.} \end{equation*} Moreover, also by Definition \ref{Definition_Pseudo_Inverse_Transformation}, \begin{equation*} P\left[\vec{Z}=\vec{z} \Bigm\vert \vec{z} \in \overset{\star}{\Omega}(\vec{Z})\right] \doteq \frac{P\left(\vec{N}=\overset{\star}{\vec{n}}\right)}{\left\lvert \overset{M}{\Omega}\left(\vec{Z}\Bigm\vert\vec{N}=\overset{\star}{\vec{n}}\right) \right\rvert} = \frac{\prod_{j=1}^{z}\overset{\star}{n}_{j}!}{n!} \textrm{.} \end{equation*} \label{Definition_Ideal_Patterning_Space} \end{definition}

\begin{definition} \textbf{Ideal Cell-Fate Proportioning Process}. An ideal or pure cell-fate proportioning process is a cell-fate proportioning process whose (sample) patterning space is, by Definition \ref{Definition_Ideal_Patterning_Space}, the ideal patterning space $\overset{\star}{\Omega}(\vec{Z})$. \label{Definition_Ideal_Cell_Fate_Proportioning_Process} \end{definition}

\begin{definition} \textbf{Uniform Patterning Space} $\overset{U}{\Omega}(\vec{Z})$. The uniform patterning space $\overset{U}{\Omega}(\vec{Z})$ is the patterning space $\Omega(\vec{Z})$ such that \begin{equation*} P(\vec{Z} = \vec{z}) = \frac{1}{\left\lvert \Omega(\vec{Z}) \right\lvert} = \frac{1}{z^{n}} \quad \forall \vec{z} \in \Omega(\vec{Z}) \textrm{.} \end{equation*} \label{Definition_Uniform_Patterning_Space} \end{definition}

\begin{definition} \textbf{Empirical Patterning Space} $\overset{E}{\Omega}(\vec{Z})$. The empirical patterning space $\overset{E}{\Omega}(\vec{Z})$ is the patterning space $\Omega(\vec{Z})$ such that all its elements are observations or measurements extracted from a given set of experiments or simulations; i.e., an empirical ensemble of developmental patterns. \label{Definition_Empirical_Patterning_Space} \end{definition}

For convenience, we will label the quantities related to the uniform, empirical, and ideal patterning spaces with the indicators $U$, $E$, and $\star$, respectively. Similarly, for convenience, we will employ an equivalent notation for the analogous counting spaces and their related quantities.

In Fig~\ref{ChapC_Fig1}, to ease the understanding of our notation, we present an illustrative example covering the ITWT system scenario where $n = 5$ and $z = 3$: Fig~\ref{ChapC_Fig1}[A] shows the complete patterning space $\Omega(\vec{Z})$; Fig~\ref{ChapC_Fig1}[B] shows the probability heatmaps associated with the uniform, empirical, and ideal patterning spaces, respectively; Fig~\ref{ChapC_Fig1}[C] shows the complete counting space $\Omega(\vec{N})$; Fig~\ref{ChapC_Fig1}[D] shows the probability heatmaps associated with the images under $T$ of the uniform, empirical, and ideal patterning spaces, respectively. Complementarily, in Appendix Fig~\ref{ApexC_Fig2}, we also present the counting vector probability heatmaps for the ITWT system scenarios where $z = 3$ and $n \in \{5,50,100\}$: we show the counting vector probability heatmaps associated with the images under $T$ of the uniform (Appendix Fig~\ref{ApexC_Fig2}[A, D, G]), empirical (Appendix Fig~\ref{ApexC_Fig2}[B, E, H]), and ideal (Appendix Fig~\ref{ApexC_Fig2}[C, F, I]) patterning spaces, respectively.

\subsection*{\texorpdfstring{An ideal cell-fate proportioning process has $\mathrm{U} = \mathrm{CI}$ ($\mathrm{PI} = 0$ always) where $\mathrm{CI}$ is a function of $n$}{An ideal cell-fate proportioning process has \textrm{U} = \textrm{CI} (\textrm{PI} = 0 always) where \textrm{CI} is a function of \textit{n}}}

Now we can formally derive the ideal versions of the proposed information-theoretical measures. Recall that an ideal cell-fate proportioning process (Definition \ref{Definition_Ideal_Cell_Fate_Proportioning_Process}) is a guiding-reference construct, which combines the uniformity assumption underlying the pseudo-inverse transformation $M$ with the assumption of a perfect counting reproducibility. Such a guiding-reference construct provides an important performance benchmark for contextualizing the application of our approach. Accordingly, for any given cell-fate proportioning process with an arbitrary number of system cells ($n$) and an arbitrary number of cell fates ($z$), we derive its ideal positional information, ideal correlational information, and ideal utility. These derivations pave the way for the last step of our strategy.

\begin{proposition} \textbf{Ideal Patterning Entropy} $\overset{\star}{\mathrm{S}}_{\mathrm{PAT}}$. \begin{equation*} \overset{\star}{\mathrm{S}}_{\mathrm{PAT}} = \log_{2}[n] - \frac{1}{n}\sum_{j=1}^{z}\overset{\star}{n}_{j}\log_{2}[\overset{\star}{n}_{j}] \textrm{.} \end{equation*} \label{Proposition_Ideal_Patterning_Entropy} \end{proposition}

\begin{proof} By Definitions \ref{Definition_Ideal_Cell_Fate_Proportioning_Process} and \ref{Definition_Patterning_Probability_Distribution}, \begin{equation*} \overset{\star}{P}_{Z=j} = \frac{\overset{\star}{n}_{j}}{n} \quad \forall j \in J_{1} \textrm{.} \end{equation*} By Definition \ref{Definition_Patterning_Entropy}, \begin{equation*} \overset{\star}{\mathrm{S}}_{\mathrm{PAT}} \doteq S\left[\overset{\star}{P}_{Z}\right] = -\sum_{j=1}^{z}\frac{\overset{\star}{n}_{j}}{n}\log_{2}\left[\frac{\overset{\star}{n}_{j}}{n}\right] \textrm{.} \end{equation*} Simplifying this expression, \begin{equation*} \begin{gathered} \overset{\star}{\mathrm{S}}_{\mathrm{PAT}} = -\frac{1}{n} \left[ \sum_{j=1}^{z}\overset{\star}{n}_{j}\log_{2}[\overset{\star}{n}_{j}] - n\log_{2}[n] \right] = \\ = \log_{2}[n] - \frac{1}{n}\sum_{j=1}^{z}\overset{\star}{n}_{j}\log_{2}[\overset{\star}{n}_{j}] \textrm{.} \end{gathered} \end{equation*} \end{proof}

\begin{proposition} \textbf{Ideal Spatial-Correlation-Free Entropy} $\overset{\star}{\mathrm{S}}_{\mathrm{SCF}}$. \begin{equation*} \overset{\star}{\mathrm{S}}_{\mathrm{SCF}} = \log_{2}[n] - \frac{1}{n}\sum_{j=1}^{z}\overset{\star}{n}_{j}\log_{2}[\overset{\star}{n}_{j}] \textrm{.} \end{equation*} \label{Proposition_Ideal_Spatial_Correlation_Free_Entropy} \end{proposition}

\begin{proof} By Definitions \ref{Definition_Ideal_Cell_Fate_Proportioning_Process} and \ref{Definition_Spatial_Correlation_Free_Probability_Distribution}, \begin{equation*} \overset{\star}{P}_{Z,N} = \begin{bmatrix} \overset{\star}{P}_{1,1} & \cdots & \overset{\star}{P}_{1,i} & \cdots & \overset{\star}{P}_{1,n} \\ \vdots &  & \vdots &  & \vdots \\ \overset{\star}{P}_{j,1} & \cdots & \overset{\star}{P}_{j,i} & \cdots & \overset{\star}{P}_{j,z} \\ \vdots &  & \vdots &  & \vdots \\ \overset{\star}{P}_{z,1} & \cdots & \overset{\star}{P}_{z,i} & \cdots & \overset{\star}{P}_{z,n} \end{bmatrix} \quad \overset{\star}{P}_{Z=j,N=i} = \frac{\overset{\star}{n}_{j}}{n} \quad \forall i \in I_{1} \wedge \forall j \in J_{1} \textrm{.} \end{equation*} By Definition \ref{Definition_Spatial_Correlation_Free_Entropy}, \begin{equation*} \begin{gathered} \overset{\star}{\mathrm{S}}_{\mathrm{SCF}} \doteq \frac{1}{n} S\left[\overset{\star}{P}_{Z,N}\right] = -\frac{1}{n} \sum_{i=1}^{n} \sum_{j=1}^{z} \overset{\star}{P}_{Z=j,N=i} \log_{2}[\overset{\star}{P}_{Z=j,N=i}] = \\ = -\frac{1}{n} \left[ \sum_{j=1}^{z}\overset{\star}{n}_{j}\log_{2}[\overset{\star}{n}_{j}] - \sum_{j=1}^{z}\overset{\star}{n}_{j}\log_{2}[n] \right] = \\ = -\frac{1}{n} \left[ \sum_{j=1}^{z}\overset{\star}{n}_{j}\log_{2}[\overset{\star}{n}_{j}] - n\log_{2}[n] \right] = \\ = \log_{2}[n] - \frac{1}{n}\sum_{j=1}^{z}\overset{\star}{n}_{j}\log_{2}[\overset{\star}{n}_{j}] \textrm{.} \end{gathered} \end{equation*} \end{proof}

\begin{adjustwidth}{0in}{0in}
\includegraphics[width = 5.5in, height = 5.5in]{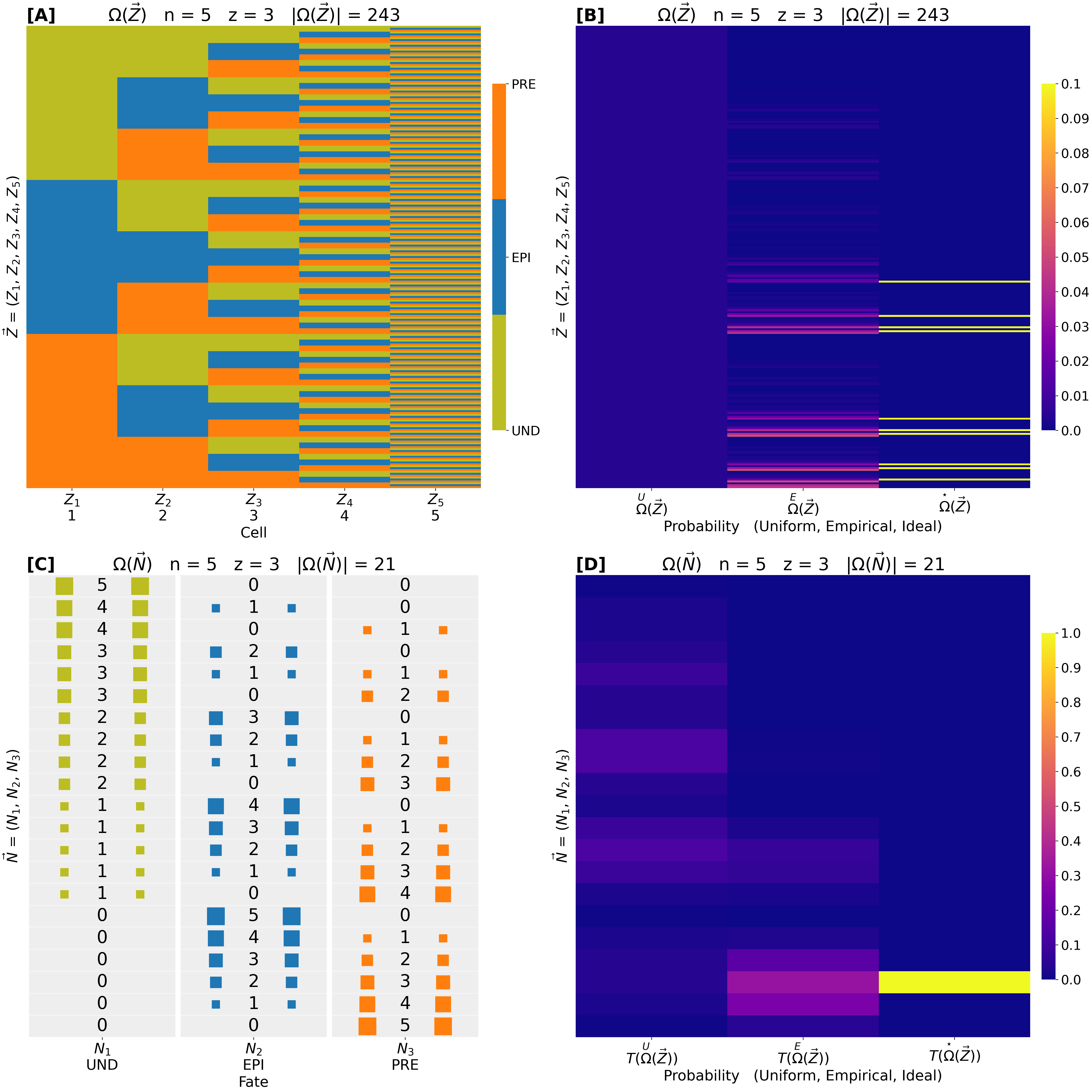} \centering 
\end{adjustwidth}
\begin{figure}[hpt!]
\caption{{\bf Illustrative example: $n = 5$ and $z = 3$.} Facilitating the understanding of our notation, we show an illustrative example for a fixed number of system cells $n = 5$ and a fixed number of cell fates $z = 3$. {\bf [A]} Patterning space $\Omega(\vec{Z})$. {\bf [B]} Probability heatmaps for uniform (left), empirical (center), and ideal (right) patterning spaces. {\bf [C]} Counting space $\Omega(\vec{N})$. {\bf [D]} Probability heatmaps for images under $T$ (Definition \ref{Definition_Transformation}) of uniform (left), empirical (center), and ideal (right) patterning spaces.}
\label{ChapC_Fig1}
\end{figure}
\clearpage

\begin{proposition} \textbf{Ideal Positional Information} $\overset{\star}{\mathrm{PI}}$. \begin{equation*} \overset{\star}{\mathrm{PI}} = 0 \textrm{.} \end{equation*} \label{Proposition_Ideal_Positional_Information} \end{proposition}

\begin{proof} By Definitions \ref{Definition_Ideal_Cell_Fate_Proportioning_Process} and \ref{Definition_Utility_Function}, $\overset{\star}{\mathrm{PI}} = \overset{\star}{\mathrm{S}}_{\mathrm{PAT}} - \overset{\star}{\mathrm{S}}_{\mathrm{SCF}}$. By Propositions \ref{Proposition_Ideal_Patterning_Entropy} and \ref{Proposition_Ideal_Spatial_Correlation_Free_Entropy}, $\overset{\star}{\mathrm{S}}_{\mathrm{PAT}} = \overset{\star}{\mathrm{S}}_{\mathrm{SCF}}$. Therefore, $\overset{\star}{\mathrm{PI}} = 0$. \end{proof}

\begin{proposition} \textbf{Ideal Reproducibility Entropy} $\overset{\star}{\mathrm{S}}_{\mathrm{REP}}$. \begin{equation*} \overset{\star}{\mathrm{S}}_{\mathrm{REP}} = \frac{1}{n}\log_{2}[n!] - \frac{1}{n}\sum_{j=1}^{z}\log_{2}[\overset{\star}{n}_{j}!] \textrm{.} \end{equation*} \label{Proposition_Ideal_Reproducibility_Entropy} \end{proposition}

\begin{proof} By Definitions \ref{Definition_Ideal_Cell_Fate_Proportioning_Process} and \ref{Definition_Reproducibility_Probability_Distribution}, \begin{equation*} \overset{\star}{P}_{\vec{Z}=\vec{z}} = \frac{\prod_{j=1}^{z}\overset{\star}{n}_{j}!}{n!} \quad \forall \vec{z} \in \overset{\star}{\Omega}(\vec{Z}) \textrm{.} \end{equation*} Recall that, by Definition \ref{Definition_Ideal_Patterning_Space}, $P(\vec{Z}=\vec{z}) = 0$ for all $\vec{z} \in \Omega(\vec{Z}) \setminus \overset{\star}{\Omega}(\vec{Z})$. By Definition \ref{Definition_Reproducibility_Entropy}, \begin{equation*} \begin{gathered} \overset{\star}{\mathrm{S}}_{\mathrm{REP}} \doteq \frac{1}{n} S\left[\overset{\star}{P}_{\vec{Z}}\right] = -\frac{1}{n} \sum_{\vec{z} \in \overset{\star}{\Omega}(\vec{Z})} \overset{\star}{P}_{\vec{Z}=\vec{z}} \log_{2}[\overset{\star}{P}_{\vec{Z}=\vec{z}}] = \\ = -\frac{1}{n} \sum_{\vec{z} \in \overset{\star}{\Omega}(\vec{Z})} \frac{\prod_{j=1}^{z}\overset{\star}{n}_{j}!}{n!} \log_{2}\left[\frac{\prod_{j=1}^{z}\overset{\star}{n}_{j}!}{n!}\right] = \\ = -\frac{1}{n} \log_{2}\left[\frac{\prod_{j=1}^{z}\overset{\star}{n}_{j}!}{n!}\right] = \\ = -\frac{1}{n} \left[ \sum_{j=1}^{z}\log_{2}[\overset{\star}{n}_{j}!] - \log_{2}[n!] \right] = \\ = \frac{1}{n}\log_{2}[n!] - \frac{1}{n}\sum_{j=1}^{z}\log_{2}[\overset{\star}{n}_{j}!] \textrm{.} \end{gathered} \end{equation*} \end{proof}

\begin{proposition} \textbf{Ideal Correlational Information} $\overset{\star}{\mathrm{CI}}$. \begin{equation*} \overset{\star}{\mathrm{CI}} = \log_{2}[n] - \frac{1}{n}\log_{2}[n!] + \frac{1}{n} \left[ \sum_{j=1}^{z}\log_{2}[\overset{\star}{n}_{j}!] - \sum_{j=1}^{z}\overset{\star}{n}_{j}\log_{2}[\overset{\star}{n}_{j}] \right] \textrm{.} \end{equation*} \label{Proposition_Ideal_Correlational_Information} \end{proposition}

\begin{proof} By Definitions \ref{Definition_Ideal_Cell_Fate_Proportioning_Process} and \ref{Definition_Utility_Function}, $\overset{\star}{\mathrm{CI}} = \overset{\star}{\mathrm{S}}_{\mathrm{SCF}} - \overset{\star}{\mathrm{S}}_{\mathrm{REP}}$. By applying Propositions \ref{Proposition_Ideal_Spatial_Correlation_Free_Entropy} and \ref{Proposition_Ideal_Reproducibility_Entropy}, we obtain the desired result. \end{proof}

\begin{proposition} \textbf{Ideal Utility} $\overset{\star}{\mathrm{U}}$. \begin{equation*} \overset{\star}{\mathrm{U}} = \log_{2}[n] - \frac{1}{n}\log_{2}[n!] + \frac{1}{n} \left[ \sum_{j=1}^{z}\log_{2}[\overset{\star}{n}_{j}!] - \sum_{j=1}^{z}\overset{\star}{n}_{j}\log_{2}[\overset{\star}{n}_{j}] \right] = \overset{\star}{\mathrm{CI}} \textrm{.} \end{equation*} \label{Proposition_Ideal_Utility} \end{proposition}

\begin{proof} By Definitions \ref{Definition_Ideal_Cell_Fate_Proportioning_Process} and \ref{Definition_Utility_Function}, $\overset{\star}{\mathrm{U}} = \overset{\star}{\mathrm{PI}} + \overset{\star}{\mathrm{CI}}$. By applying Propositions \ref{Proposition_Ideal_Positional_Information} and \ref{Proposition_Ideal_Correlational_Information}, we obtain the desired result. \end{proof}

The final step of our $\mathrm{S}_{\mathrm{REP}}$ estimation strategy requires gathering knowledge about the underlying counting space of a given cell-fate proportioning process. This final step is flexible: it can exploit experimental or synthetic (simulation) data to approximate the underlying counting probability space.

To exemplify our complete strategy, we employ a concrete simulation dataset. In \cite{ramirez-sierra_ai-powered_2024}, we generated simulation data for the analysis of multiple ITWT system scenarios corresponding to 13 distinct inner cell mass ``ICM'' sizes (number of system cells): $n \in \{5,10,15,25,35,50,65,75,85,100,150,225,400\}$; recall that $z = 3$. We leveraged this simulation dataset to compute the empirical patterning, spatial-correlation-free, and reproducibility entropies for all the available ITWT system scenarios. For each ICM size, to compute the empirical reproducibility entropy, we created a list of all the realizable elements of the respective (sample) counting space, estimating their probabilities based on the observed counting vectors. Following this estimation of the counting probability space, we applied Definition \ref{Definition_Pseudo_Inverse_Transformation} (pseudo-inverse transform $M$) for assigning probability estimates to patterning vector sets, thus obtaining an approximation of the respective patterning probability space.

As seen in Fig~\ref{ChapC_Fig2}[A, B, C], we compare the empirical (solid lines) patterning, spatial-correlation-free, and reproducibility entropies against their ideal (dashed lines) analogs. For the empirical case, we observe that $\mathrm{S}_{\mathrm{PAT}}$ (dark orchid), $\mathrm{S}_{\mathrm{SCF}}$ (slate gray), and $\mathrm{S}_{\mathrm{REP}}$ (dark turquoise) are almost equal irrespective of $n$. For the ideal case, we observe that, while $\mathrm{S}_{\mathrm{PAT}}$ and $\mathrm{S}_{\mathrm{SCF}}$ are constant with respect to $n$, $\mathrm{S}_{\mathrm{REP}}$ rapidly approaches its upper asymptote with increasing $n$ (see Fig~\ref{ChapC_Fig2}[C]).

In Fig~\ref{ChapC_Fig2}[D], we also compare the positional-correlational information measures alongside the utility function for the empirical and ideal cases. As proposed for the ideal case, $\mathrm{PI} = 0$ irrespective of $n$, while $\mathrm{U} = \mathrm{CI}$ rapidly approaches its lower asymptote with increasing $n$. The observed shape of the ideal correlational information (with respect to $n$) indicates that the ideal utility is a strictly monotonically decreasing function of $n$. For the empirical case, although the utility is low and rapidly diminishes with growing $n$ (it ranges from around 0.007 to around 0.001 bits), we observe that it never fully vanishes. This finding is interesting because such a low utility suggests that the ITWT system exhibits a behavior which is nowhere near optimal. However, this finding partially counters the results of \cite{ramirez-sierra_ai-powered_2024}, where we identified for the ITWT system a highly accurate and precise behavior, revealing a robust and reproducible cell-fate proportioning process. Moreover, we also determined that the precision of the cell-fate distributions increases with growing $n$, whereas their accuracy remains constant with respect to $n$. Altogether, we propose that the current view of optimal behavior for cell-fate proportioning processes is incomplete, and it is necessary to reinterpret optimality within this context; see Discussion section for details.

Complementarily, we calculated the time series reflecting the temporal evolution of all the empirical quantities ($\mathrm{S}_{\mathrm{PAT}}$, $\mathrm{S}_{\mathrm{SCF}}$, $\mathrm{S}_{\mathrm{REP}}$, $\mathrm{PI}$, $\mathrm{CI}$, and $\mathrm{U}$) for the ITWT system scenarios corresponding to $n = 5$ (see Fig~\ref{ChapC_Fig3}) as well as $n = 25$ (see Appendix Fig~\ref{ApexC_Fig3}); we kept their ideal analogs as guiding references. We observe that all the considered empirical quantities are almost equal irrespective of the time point: there is a rapid increase of entropy caused by a quick surge of the EPI cell population, which is followed by a slow but steady entropy increase due to the gradual specification of the PRE cell population. Moreover, the maximum entropy is reached around the 24-hour mark, and after this peak we observe a progressive entropy decrease until the last available time point (see panels [A, B, C] of Fig~\ref{ChapC_Fig3} and Appendix Fig~\ref{ApexC_Fig3}). This observation suggests that, during the last half of the target developmental period, the UND cell population experiences a refinement until it almost vanishes. The empirical patterning entropy finalizes with approximately 0.2 ($n = 5$) and 0.1 ($n = 25$) bits more (better) than its ideal analog, but the empirical reproducibility entropy finalizes with approximately 0.5 ($n = 5$) and 0.25 ($n = 25$) bits more (worse) than its ideal analog. As shown in panel [D] of Fig~\ref{ChapC_Fig3} and Appendix Fig~\ref{ApexC_Fig3}, the empirical utility increases over time, but it only achieves around a fifth of the ideal utility (for both considered ITWT system scenarios).

\begin{adjustwidth}{0in}{0in}
\includegraphics[width = 5.25in, height = 5.25in]{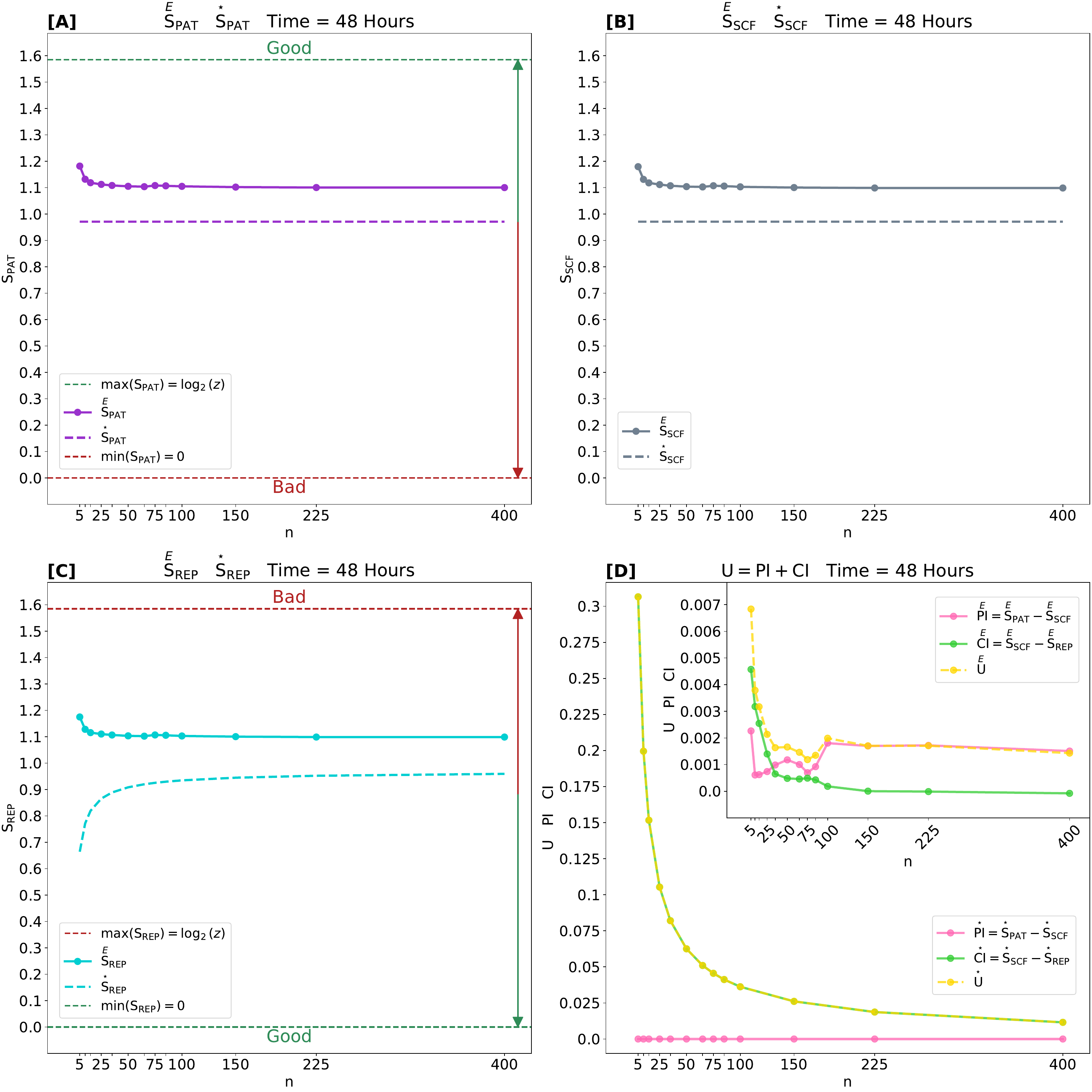} \centering 
\end{adjustwidth}
\begin{figure}[hpt!]
\caption{{\bf Comparing empirical information-theoretical measures against their ideal analogs (with respect to ICM size): ITWT system.} Time = 48 Hours. Note that green (good) and red (bad) signs indicate better and worse performance, respectively. {\bf [A]} Measure of patterning entropy. {\bf [B]} Measure of spatial-correlation-free entropy. {\bf [C]} Measure of reproducibility entropy. {\bf [D]} Measure of utility: positional information plus correlational information. The main plot shows the ideal measures. The inset shows the empirical measures.}
\label{ChapC_Fig2}
\end{figure}

\begin{adjustwidth}{0in}{0in}
\includegraphics[width = 5.9in, height = 5.9in]{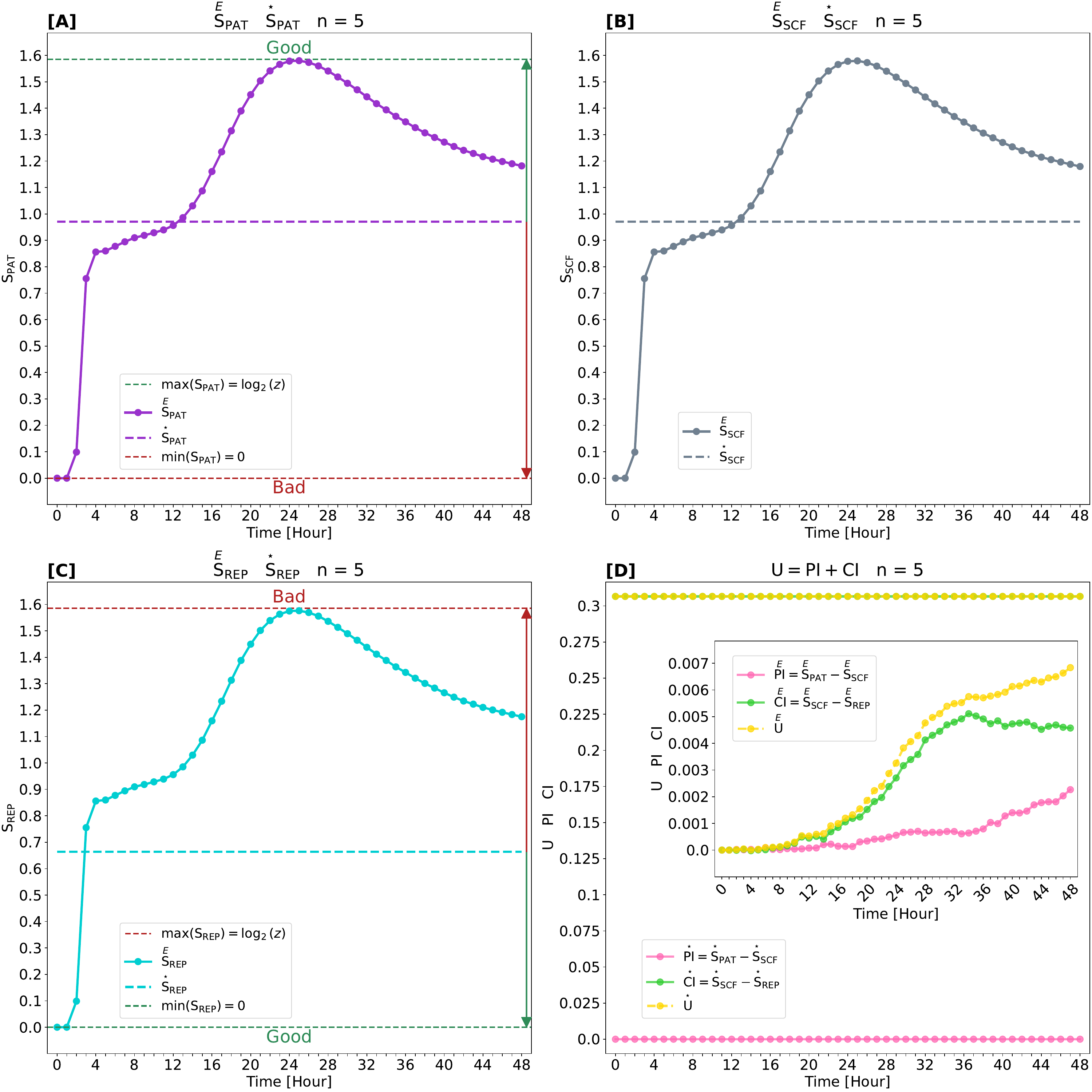} \centering 
\end{adjustwidth}
\begin{figure}[hpt!]
\caption{{\bf Comparing empirical information-theoretical measures against their ideal analogs (with respect to simulation time point): ITWT system.} Number of system cells $n = 5$. Note that green (good) and red (bad) signs indicate better and worse performance, respectively. {\bf [A]} Measure of patterning entropy. {\bf [B]} Measure of spatial-correlation-free entropy. {\bf [C]} Measure of reproducibility entropy. {\bf [D]} Measure of utility: positional information plus correlational information. The main plot shows the ideal measures. The inset shows the empirical measures.}
\label{ChapC_Fig3}
\end{figure}

\subsection*{\texorpdfstring{An alternative measure for quantifying reproducibility entropy ($\mathrm{W}_{\mathrm{REP}}$)}{An alternative measure for quantifying reproducibility entropy (\textrm{W}\_{\textrm{REP}})}}

Arguably, for any given cell-fate proportioning process (which is not necessarily ideal in the sense portrayed by our work), the main objective is to generate cell-fate counting distributions that are not only consistent across replicates of developmental pattern ensembles, but also reproducible under different environmental conditions and robust to biological noise or systematic perturbations. As such, we believe that it is advantageous to consider a (simple) alternative measure for quantifying the information content of a cell-fate proportioning process. This proposal for an alternative measure is formalized in the following definition.

\begin{definition} \textbf{Reproducibility Entropy (Counting Space)} $\mathrm{W}_{\mathrm{REP}}$. Let \begin{equation*} P_{\vec{N}} \doteq P(\vec{N}=\vec{n}) \quad \forall \vec{n} \in \Omega(\vec{N}) \textrm{.} \end{equation*} The counting probability distribution $P_{\vec{N}}$ dictates the probability of observing realization (cell-fate count) $\vec{n} \in \Omega(\vec{N})$. For convenience, $P_{\vec{N}}$ is equivalent to $P(\vec{N})$; i.e., $P_{\vec{N}} \equiv P(\vec{N})$. Let \begin{equation*} \mathrm{W}_{\mathrm{REP}} \doteq \frac{1}{n} S\left[P_{\vec{N}}\right] = -\frac{1}{n} \sum_{\vec{n} \in \Omega(\vec{N})} P(\vec{N}=\vec{n}) \log_{2}[P(\vec{N}=\vec{n})] \textrm{.} \end{equation*} In words, $\mathrm{W}_{\mathrm{REP}}$ measures the cell-fate counting distribution reproducibility over an ensemble of developmental patterns. It is easy to see that, if $\exists ! \vec{n} \in \Omega(\vec{N})$ such that $P(\vec{N}=\vec{n}) = 1$, then $\min(\mathrm{W}_{\mathrm{REP}}) = 0$. Likewise, if $P(\vec{N}=\vec{n}) = 2/((n+1)(n+2))$ for all $\vec{n} \in \Omega(\vec{N})$, then $\max(\mathrm{S}_{\mathrm{REP}}) = (\log_{2}[n+1]+\log_{2}[n+2]-1)/n$. \label{Definition_Reproducibility_Entropy_Counting_Space} \end{definition}

By applying Definition \ref{Definition_Reproducibility_Entropy_Counting_Space} to all the distinct ITWT system scenarios described previously, 13 distinct ICM sizes, we obtained quantifications of the reproducibility entropy $\mathrm{W}_{\mathrm{REP}}$ for five different cases: (1) a counting space which is the image under $T$ of the uniform patterning space; (2) a uniform counting space; (3) a counting space which is constructed using a collection of samples from a multinomial distribution; (4) an empirical counting space; (5) an ideal counting space. Cases one and two, $T(\overset{U}{\Omega}(\vec{Z}))$ and $\overset{U}{\Omega}(\vec{N})$, are guiding references. Case three, $\overset{K}{\Omega}(\vec{N})$, represents an artificially generated system producing cell-fate counting distributions where the cell-fate counts are sampled from a multinomial distribution with exactly the same proportions as the observed simulation data (for all the simulated times). Case four, $\overset{E}{\Omega}(\vec{N})$, is our prime purpose. Case five, $\overset{\star}{\Omega}(\vec{N})$, is a trivial guiding reference.

In Fig~\ref{ChapC_Fig4}, we see that, for all the five different $\mathrm{W}_{\mathrm{REP}}$ cases, the entropy decreases (improves) consistently with growing $n$ (ITWT system scenario or size). In particular for the artificial counting probability space associated with a multinomial distribution, its counting reproducibility entropy is systematically higher (worse) than its empirical counterpart, although this difference is modest: the maximum gap is around 0.05 bits when $n = 5$, corresponding to approximately 5\% of the highest entropy which is dictated by the uniform counting space. This finding recapitulates a result of \cite{ramirez-sierra_ai-powered_2024}, where we identified that a multinomial formalism for the so-called ``salt-and-pepper'' patterning process produces cell-fate counting distributions almost identical to the simulated ITWT system behavior realizations. Focusing on the maximum relative difference (when $n = 5$) between the multinomial and empirical cases, we find that the empirical counting reproducibility entropy is lower (better) than its multinomial counterpart by around 10\%. This finding parallels another result of \cite{ramirez-sierra_ai-powered_2024}, where we identified that the cell-fate counting process for the ITWT system exhibits a coefficient of variation which is around 10\% lower (better) than the expected measure for a cell-fate decisioning process with a simple binomial noise assumption, regardless of ICM size. Altogether, we believe that the counting reproducibility entropy $\mathrm{W}_{\mathrm{REP}}$ aligns with the optimality intuitions emerging from \cite{ramirez-sierra_ai-powered_2024} and \cite{ramirez_sierra_comparing_2025} accompanying this study.

Moreover, as shown in Fig~\ref{ChapC_Fig5} ($n = 5$) and Appendix Fig~\ref{ApexC_Fig4} ($n = 25$), we also calculated the temporal dynamics for all the five different preceding cases. Overall, we see that the empirical counting reproducibility entropy $\mathrm{W}_{\mathrm{REP}}$ tracks reasonably well the empirical patterning reproducibility entropy $\mathrm{S}_{\mathrm{REP}}$ (Fig~\ref{ChapC_Fig3} and Appendix Fig~\ref{ApexC_Fig3}): both empirical quantities exhibit a similar time series shape, and appear to be rescaled versions of one another. Interestingly, a clear bound arises from this analysis: neither multinomial nor empirical quantities surpass the entropy dictated by a counting space which is the image under $T$ of the uniform patterning space. This observation can be explained by thinking in terms of Definitions \ref{Definition_Transformation} and \ref{Definition_Pseudo_Inverse_Transformation}. In this way, to generate a uniform counting space, any given cell-fate proportioning process must bias its patterning space such that the realized developmental pattern ensembles produce uniformly distributed cell-fate counting vectors. But such a biased process, in view of Definitions \ref{Definition_Transformation} and \ref{Definition_Pseudo_Inverse_Transformation}, cannot properly reflect any natural cell-fate proportioning mechanism because it will imply that the underlying patterning and counting space structures are known to the developmental system. In other words, given that the cardinality of the set of patterning vectors generating a particular counting vector is dictated by each unique realizable counting vector, no natural cell-fate proportioning process can realistically produce a uniform counting space. As such, the biologically feasible maximum counting entropy corresponds with the maximum patterning entropy.

\begin{adjustwidth}{0in}{0in}
\includegraphics[width = 5.25in, height = 5.25in]{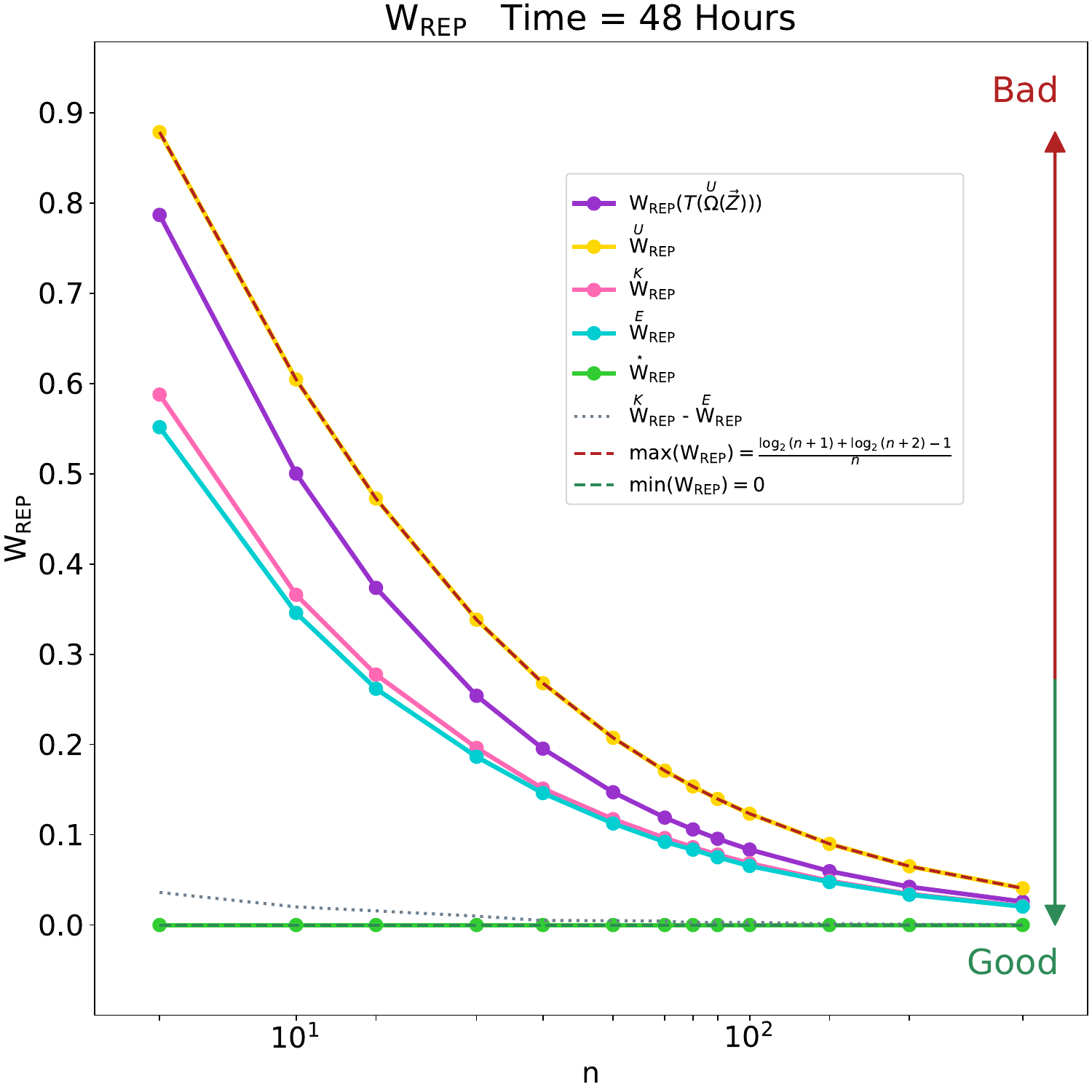} \centering 
\end{adjustwidth}
\begin{figure}[hpt!]
\caption{{\bf Reproducibility entropy (counting space) with respect to ICM size: $\mathrm{W}_{\mathrm{REP}}$.} Time = 48 Hours. Quantifications of reproducibility entropy $\mathrm{W}_{\mathrm{REP}}$ for five different cases: (dark orchid) image under $T$ (Definition \ref{Definition_Transformation}) of uniform patterning space; (gold) uniform counting space; (hot pink) multinomial counting space; (dark turquoise) empirical counting space; (lime green) ideal counting space. Note that green (good) and red (bad) signs indicate better and worse performance, respectively. The dotted (slate gray) line represents the difference between multinomial and empirical counting cases. The dashed (dark red and dark green) lines represent the maximum and minimum entropies, respectively.}
\label{ChapC_Fig4}
\end{figure}
\clearpage

\begin{adjustwidth}{0in}{0in}
\includegraphics[width = 5.9in, height = 5.9in]{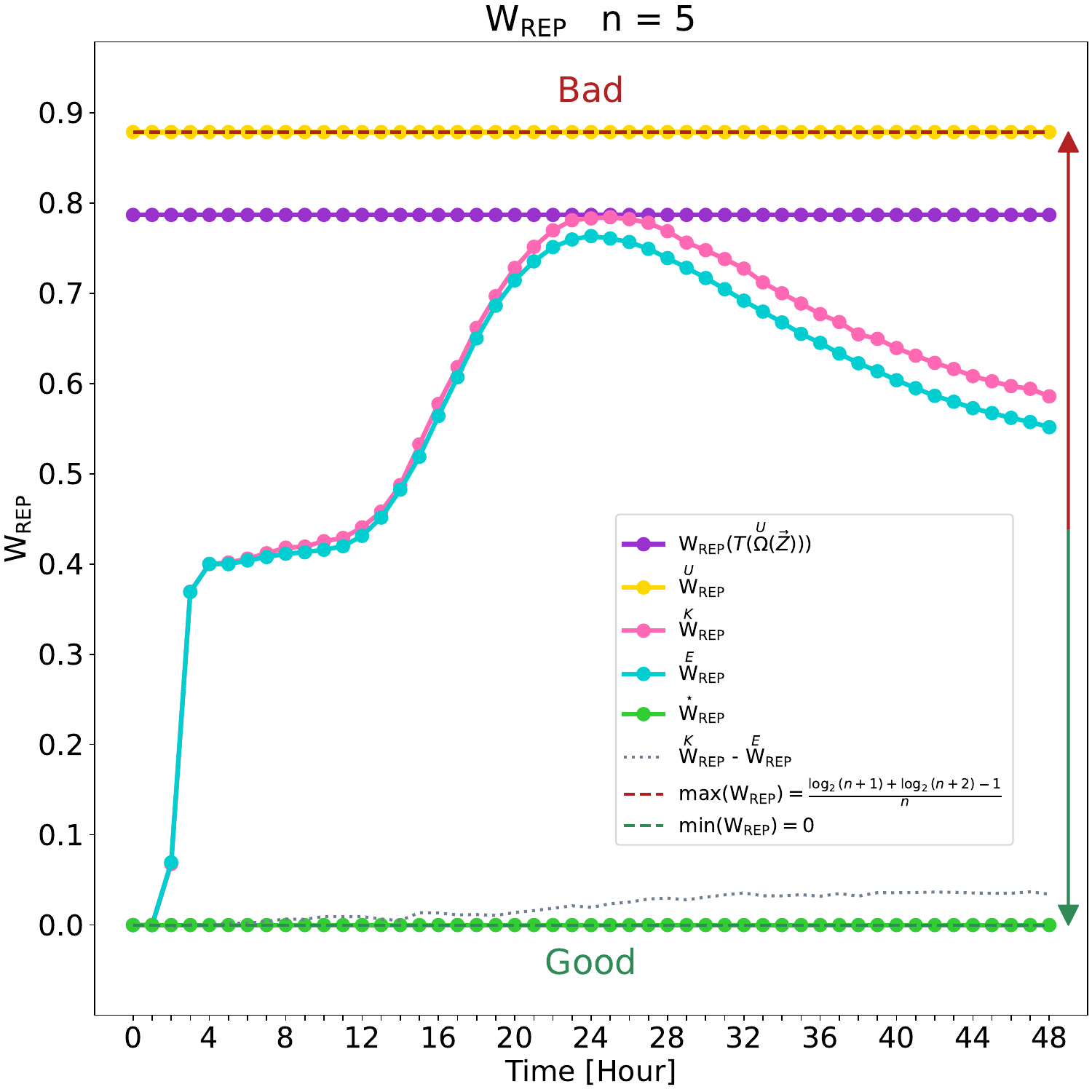} \centering 
\end{adjustwidth}
\begin{figure}[hpt!]
\caption{{\bf Reproducibility entropy (counting space) with respect to simulation time point: $\mathrm{W}_{\mathrm{REP}}$.} Number of system cells $n = 5$. Quantifications of reproducibility entropy $\mathrm{W}_{\mathrm{REP}}$ for five different cases: (dark orchid) image under $T$ (Definition \ref{Definition_Transformation}) of uniform patterning space; (gold) uniform counting space; (hot pink) multinomial counting space; (dark turquoise) empirical counting space; (lime green) ideal counting space. Note that green (good) and red (bad) signs indicate better and worse performance, respectively. The dotted (slate gray) line represents the difference between multinomial and empirical counting cases. The dashed (dark red and dark green) lines represent the maximum and minimum entropies, respectively.}
\label{ChapC_Fig5}
\end{figure}
\clearpage




\section*{Discussion}

Developmental processes such as the blastocyst formation are fascinating and paradigmatic examples of self-organization \cite{bruckner_information_2024, plusa_common_2020, schroter_local_2023, zhu_principles_2020, morales_embryos_2021}. These self-organized processes are capable not only of generating highly reproducible and robust cell-fate specification behavior starting from homogeneous cellular populations (while under noisy biological conditions) \cite{schroter_local_2023}, but also of synergistically maintaining a high cellular self-renewal capacity over restricted developmental periods \cite{plusa_common_2020}. As such, a formal-yet-generic notion of self-organization is fundamentally required for uncovering the general principles that govern the emergence of complexity for developmental systems. However, until recently, there was no universal frame of reference for quantifying the self-organization capability (information content) of developmental biology systems like the blastocyst; any available option was tailored to the qualitative properties of a particular developmental mechanism.

In this study, applying a newly proposed information-theoretical framework for quantifying the self-organization potential of cell-fate patterning processes \cite{bruckner_information_2024}, we studied the information content of our simulated Inferred-Theoretical Wild-Type ``ITWT'' system (see \cite{ramirez-sierra_ai-powered_2024}). More explicitly, we devised a mathematical and computational strategy for calculating the patterning (Definition \ref{Definition_Patterning_Entropy}), spatial-correlation-free (Definition \ref{Definition_Spatial_Correlation_Free_Entropy}), and reproducibility (Definition \ref{Definition_Reproducibility_Entropy}) entropies, which are proposed in the study by Brückner and Tkačik \cite{bruckner_information_2024}, for our cell-fate proportioning (ITWT) system. Importantly, this strategy overcomes the most important difficulty of the newly proposed quantification framework: the calculation of the reproducibility entropy necessitates the estimation of a high dimensional probability space, whose cardinality grows exponentially as a function of $n$ (number of system cells) and $z$ (number of cell fates): $z^{n}$.

By creating an (indirect) mapping from cell-fate counting probability space to cell-fate patterning probability space, we estimated the self-organization potential of our ITWT system. Within this framework \cite{bruckner_information_2024}, we found that our cell-fate proportioning system exhibits a low empirical utility (see Definition \ref{Definition_Utility_Function}), which comparatively only reaches approximately a fifth of the ideal utility (see Definition \ref{Definition_Ideal_Cell_Fate_Proportioning_Process}). This finding partially contradicts the intuitions and results of \cite{ramirez-sierra_ai-powered_2024}, where we observed that the ITWT system generates highly accurate cell-fate counting distributions whose precision increases with growing $n$ ($z = 3$). This contradiction suggests that, to complete our view of cell-fate proportioning processes, it should be beneficial to alternatively (or complementarily) consider an entropy measure for the available counting space. Instead of directly assessing the information content of cell-fate patterning distributions for any given ensemble of developmental patterns, we should quantify the entropy of the realized cell-fate counting vectors. To this end, we proposed a simple measure of reproducibility entropy for any given ensemble of developmental (cell-fate) counts; see Defintion \ref{Definition_Reproducibility_Entropy_Counting_Space}. This alternative measure consistently recapitulates our previous findings where we observed a highly reproducible and robust cell-fate specification behavior for the ITWT system: the alternative reproducibility entropy decreases (improves) with growing $n$ ($z = 3$), and its empirical time series shows that the ITWT system performs better (although marginally) than an artificial (less complex) system whose sample space is dictated by a multinomial distribution; i.e., an artificial system without cell-cell signaling or communication.

Although our measure of counting reproducibility is simple and computationally tractable, it should be viewed not only as an alternative to but also as an extension of the framework recently proposed by Brückner and Tkačik \cite{bruckner_information_2024}. The formal relationships that we have constructed between patterning spaces and counting spaces lay the foundation for improving our understanding of cell-fate spatial patterning and cell-fate count proportioning from multiple perspectives. While these processes might appear equivalent, they actually reflect two complementary aspects of the same developmental mechanism: cell-fate spatial patterning targets the timely formation of a spatially coordinated arrangement of cell fates without imposing count constraints; cell-fate count proportioning targets the timely formation of a replicable counting ratio of cell fates without imposing spatial constraints. Recognizing this distinction is crucial for appropriately contextualizing both experimental and simulation results, as these complementary processes demand distinct conditions for proper characterization. More explicitly, the cardinality of the patterning space is often significantly greater than the cardinality of its corresponding counting space, except for trivial cases. This difference suggests that efficient approximations of patterning probability spaces could be derived from available data on counting spaces, thereby accelerating the computation of the proposed information-theoretical measures, as demonstrated by our work. These approximations could be refined by incorporating rich structural constraints within the underlying spaces and establishing strategies to filter biologically feasible cell-fate pattern distributions based on observed cell-fate counting distributions.

This broad information-theoretical framework for quantifying self-organization potential has a promisingly high applicability across a vast range of biological systems. It offers a universal and generic approach to evaluate the inherent functional roles played by the main components of any developmental mechanism, including interactions within gene regulatory network motifs. By formalizing the criteria for self-organization and adopting such a versatile framework, we could enhance the identification of optimal parameter regimes and advance the creation of highly detailed mechanistic models, which are essential for disentangling the key drivers of correct development.





\bibliography{Manuscript_Catalog.bib}

\begin{thebibliography}{10}

\bibitem{bruckner_information_2024}
Brückner DB, Tkačik G.
\newblock Information content and optimization of self-organized developmental
  systems.
\newblock Proceedings of the National Academy of Sciences.
  2024;121(23):e2322326121.
\newblock doi:{10.1073/pnas.2322326121}.

\bibitem{plusa_common_2020}
Płusa B, Piliszek A.
\newblock Common principles of early mammalian embryo self-organisation.
\newblock Development. 2020;147(dev183079).
\newblock doi:{10.1242/dev.183079}.

\bibitem{schroter_local_2023}
Schröter C, Stapornwongkul KS, Trivedi V.
\newblock Local cellular interactions during the self-organization of stem
  cells.
\newblock Current Opinion in Cell Biology. 2023;85:102261.
\newblock doi:{10.1016/j.ceb.2023.102261}.

\bibitem{zhu_principles_2020}
Zhu M, Zernicka-Goetz M.
\newblock Principles of {Self}-{Organization} of the {Mammalian} {Embryo}.
\newblock Cell. 2020;183(6):1467--1478.
\newblock doi:{10.1016/j.cell.2020.11.003}.

\bibitem{morales_embryos_2021}
Morales JS, Raspopovic J, Marcon L.
\newblock From embryos to embryoids: {How} external signals and
  self-organization drive embryonic development.
\newblock Stem Cell Reports. 2021;16(5):1039--1050.
\newblock doi:{10.1016/j.stemcr.2021.03.026}.

\bibitem{beck_understanding_2024}
Beck M, Covino R, Hänelt I, Müller-McNicoll M.
\newblock Understanding the cell: {Future} views of structural biology.
\newblock Cell. 2024;187(3):545--562.
\newblock doi:{10.1016/j.cell.2023.12.017}.

\bibitem{ryan_lumen_2019}
Ryan AQ, Chan CJ, Graner F, Hiiragi T.
\newblock Lumen {Expansion} {Facilitates} {Epiblast}-{Primitive} {Endoderm}
  {Fate} {Specification} during {Mouse} {Blastocyst} {Formation}.
\newblock Developmental Cell. 2019;51(6):684--697.e4.
\newblock doi:{10.1016/j.devcel.2019.10.011}.

\bibitem{ramirez-sierra_ai-powered_2024}
Ramirez~Sierra MA, Sokolowski TR.
\newblock {AI}-powered simulation-based inference of a genuinely
  spatial-stochastic gene regulation model of early mouse embryogenesis.
\newblock PLOS Computational Biology. 2024;20(11):e1012473.
\newblock doi:{10.1371/journal.pcbi.1012473}.

\bibitem{zagorski_decoding_2017}
Zagorski M, Tabata Y, Brandenberg N, Lutolf MP, Tkačik G, Bollenbach T, et~al.
\newblock Decoding of position in the developing neural tube from antiparallel
  morphogen gradients.
\newblock Science. 2017;356(6345):1379--1383.
\newblock doi:{10.1126/science.aam5887}.

\bibitem{tkacik_many_2021}
Tkačik G, Gregor T.
\newblock The many bits of positional information.
\newblock Development. 2021;148(dev176065).
\newblock doi:{10.1242/dev.176065}.

\bibitem{seyboldt_latent_2022}
Seyboldt R, Lavoie J, Henry A, Vanaret J, Petkova MD, Gregor T, et~al.
\newblock Latent space of a small genetic network: {Geometry} of dynamics and
  information.
\newblock Proceedings of the National Academy of Sciences.
  2022;119(26):e2113651119.
\newblock doi:{10.1073/pnas.2113651119}.

\bibitem{sokolowski_deriving_2023}
Sokolowski TR, Gregor T, Bialek W, Tkačik G. Deriving a genetic regulatory
  network from an optimization principle; 2023.
\newblock Available from: \url{http://arxiv.org/abs/2302.05680}.

\bibitem{ramirez_sierra_comparing_2025}
Ramirez~Sierra MA, Sokolowski TR.
\newblock Comparing {AI} versus optimization workflows for simulation-based
  inference of spatial-stochastic systems.
\newblock Machine Learning: Science and Technology. 2025;6(1):010502.
\newblock doi:{10.1088/2632-2153/ada0a3}.

\end{thebibliography}

\newpage


\section*{Acknowledgments}

The successful completion of this research project owes much to the collaboration and support of esteemed colleagues and collaborators. We extend our deepest gratitude to Gašper Tkačik and David B. Brückner for their pivotal roles in fostering insightful discussions and providing constructive feedback. Additionally, we acknowledge the financial support from the CMMS project and FIAS, which has been instrumental for the advancement of this research. We also thank the Center for Scientific Computing (CSC) at Goethe University Frankfurt for providing access to the Goethe-HLR cluster.

\section*{Funding Information}

This research work was funded by the LOEWE-Schwerpunkt ``Center for Multiscale Modelling in Life Sciences'' (CMMS), which is sponsored by the Hessian Ministry of Science and Research, Arts and Culture{\textemdash}Hessisches Ministerium für Wissenschaft und Forschung, Kunst und Kultur{\textemdash}(HMWK). The funders had no role in study design, data collection and analysis, decision to publish, or preparation of the manuscript.

\section*{Data Availability}

All code files will be available from a GitHub repository. All estimated parameter posterior distributions will be available from the same GitHub repository. The complete data bank will be available from a Zenodo repository.

\section*{Competing Interests}

The authors declare that no competing interests exist.


\newpage


\section*{Supporting information} \label{ApexC} 





\begin{adjustwidth}{0in}{0in}
\includegraphics[width = 5.5in, height = 5.5in]{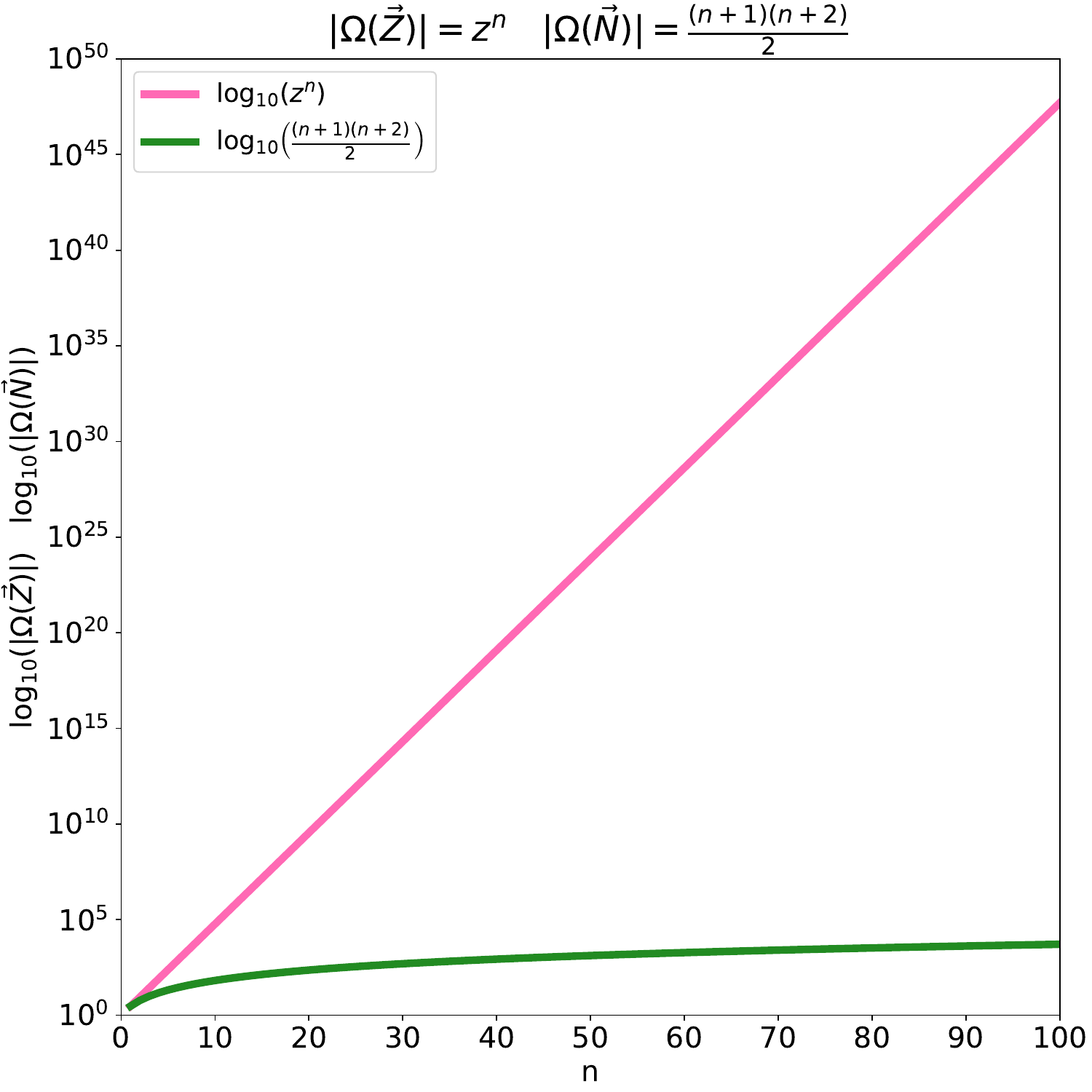} \centering 
\end{adjustwidth}
\begin{figure}[hpt!]
\caption{{\bf Comparison between patterning $\Omega(\vec{Z})$ and counting $\Omega(\vec{N})$ space cardinalities.} Here, if $n \gg 1$, then $((n+1)(n+2))/2 \ll z^{n}$. Notation: number of system cells $n$; number of cell fates $z$.}
\label{ApexC_Fig1}
\end{figure}
\clearpage

\begin{adjustwidth}{0in}{0in}
\includegraphics[width = 5.9in, height = 5.9in]{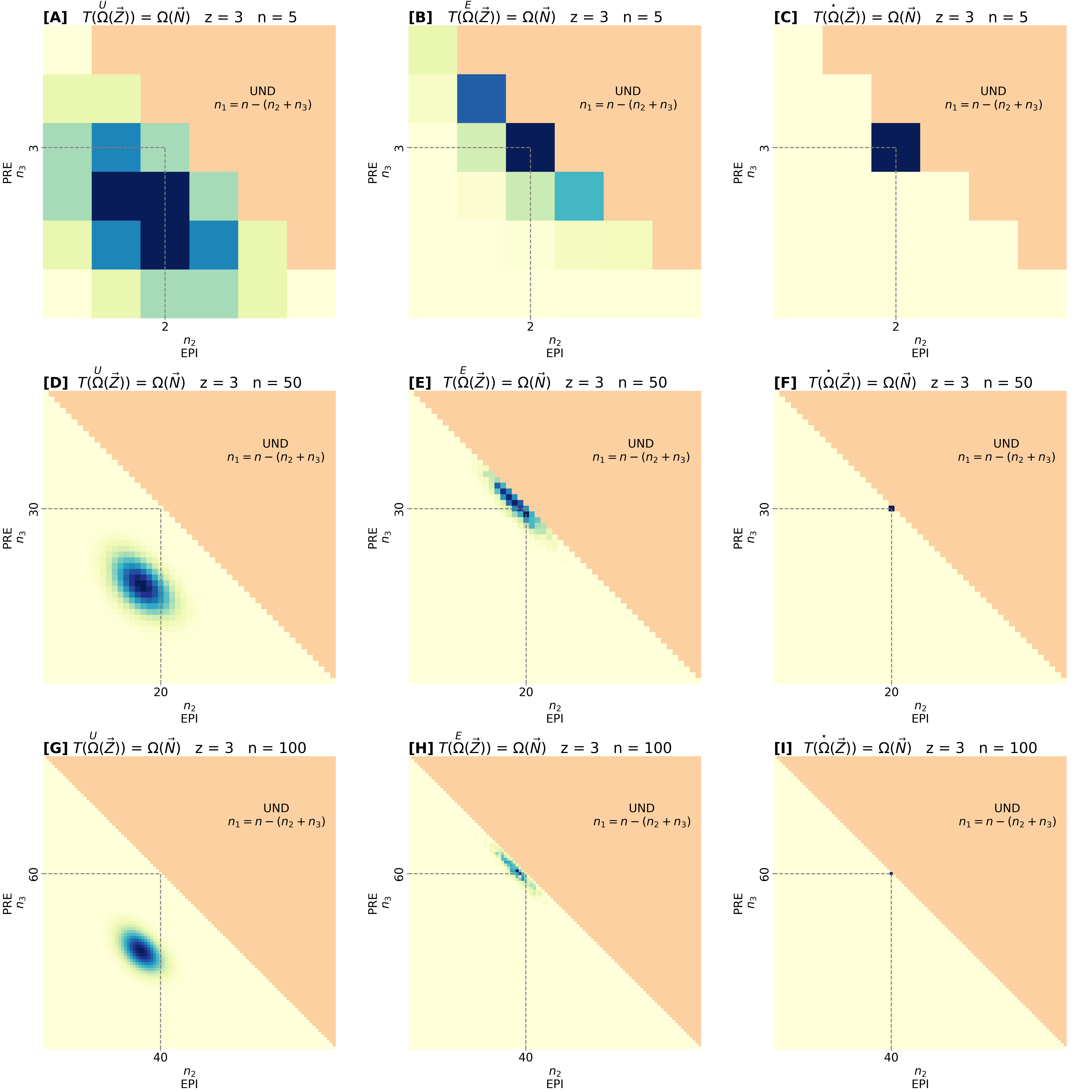} \centering 
\end{adjustwidth}
\begin{figure}[hpt!]
\caption{{\bf Counting vector probability heatmaps: ITWT system.} Facilitating the understanding of our notation, we present an illustrative example for the ITWT system scenarios where $z = 3$ and $n \in \{5,50,100\}$. Here, $n$ is the number of system cells and $z$ is the number of cell fates. {\bf [A-C]} Scenario: $n = 5$; $z = 3$. {\bf [D-F]} Scenario: $n = 50$; $z = 3$. {\bf [G-I]} Scenario: $n = 100$; $z = 3$. {\bf [A, D, G]} Counting vector probability heatmaps for images under $T$ (Definition \ref{Definition_Transformation}) of uniform patterning spaces. {\bf [B, E, H]} Counting vector probability heatmaps for images under $T$ of empirical patterning spaces. {\bf [C, F, I]} Counting vector probability heatmaps for images under $T$ of ideal patterning spaces.}
\label{ApexC_Fig2}
\end{figure}
\clearpage

\begin{adjustwidth}{0in}{0in}
\includegraphics[width = 5.9in, height = 5.9in]{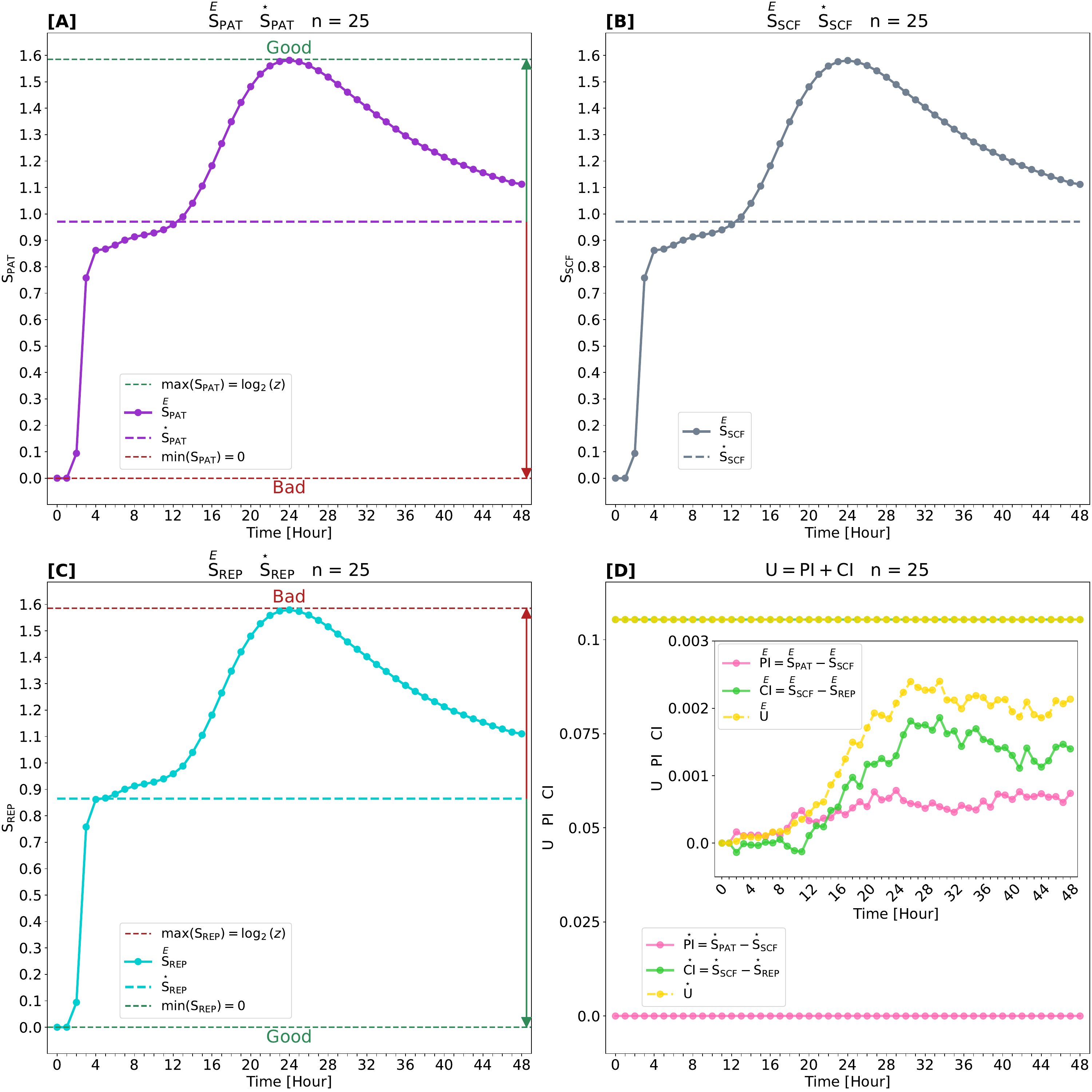} \centering 
\end{adjustwidth}
\begin{figure}[hpt!]
\caption{{\bf Comparing empirical information-theoretical measures against their ideal analogs (with respect to simulation time point): ITWT system.} Number of system cells $n = 25$. Note that green (good) and red (bad) signs indicate better and worse performance, respectively. {\bf [A]} Measure of patterning entropy. {\bf [B]} Measure of spatial-correlation-free entropy. {\bf [C]} Measure of reproducibility entropy. {\bf [D]} Measure of utility: positional information plus correlational information. The main plot shows the ideal measures. The inset shows the empirical measures.}
\label{ApexC_Fig3}
\end{figure}
\clearpage

\begin{adjustwidth}{0in}{0in}
\includegraphics[width = 5.9in, height = 5.975in]{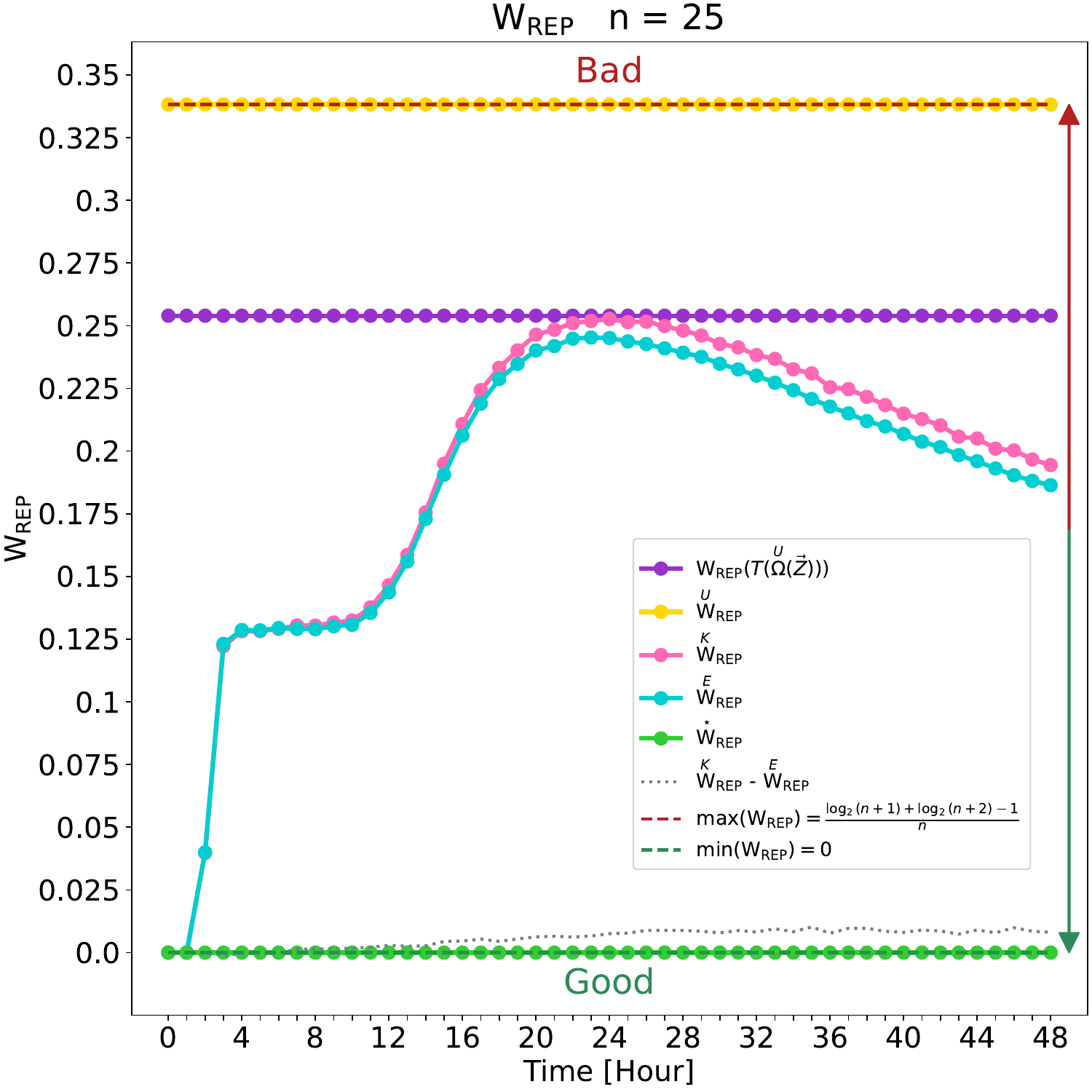} \centering 
\end{adjustwidth}
\begin{figure}[hpt!]
\caption{{\bf Reproducibility entropy (counting space) with respect to simulation time point: $\mathrm{W}_{\mathrm{REP}}$.} Number of system cells $n = 25$. Quantifications of reproducibility entropy $\mathrm{W}_{\mathrm{REP}}$ for five different cases: (dark orchid) image under $T$ of uniform patterning space; (gold) uniform counting space; (hot pink) multinomial counting space; (dark turquoise) empirical counting space; (lime green) ideal counting space. Note that green (good) and red (bad) signs indicate better and worse performance, respectively. The dotted (slate gray) line represents the difference between multinomial and empirical counting cases. The dashed (dark red and dark green) lines represent the maximum and minimum entropies, respectively.}
\label{ApexC_Fig4}
\end{figure}
\clearpage


\end{document}